\documentclass[a4paper,UKenglish]{lipics-v2016}
 
\usepackage{microtype}

\usepackage{amsmath,graphicx,epsfig,color,amsfonts,amssymb,bbm}
\usepackage{amsthm}
\usepackage{mathtools}

\usepackage{changes}
\usepackage{times,color}

\definecolor{myurlcolor}{rgb}{0,0,0.7}
\definecolor{myrefcolor}{rgb}{0.8,0,0}
\usepackage{hyperref}
\hypersetup{colorlinks, linkcolor=myrefcolor,
citecolor=myurlcolor, urlcolor=myurlcolor}

\newcommand{\ket}[1]{|#1\rangle}

\definecolor{nblue}{rgb}{0.2,0.2,0.7}
\definecolor{ngreen}{rgb}{0.2,0.6,0.2}
\definecolor{nred}{rgb}{0.8,0.2,0.2}
\definecolor{nblack}{rgb}{0,0,0}

\newtheorem{proposition}{Proposition}

\title{Entangled systems are unbounded sources of nonlocal correlations and of certified random numbers}
\titlerunning{Entangled system are unbounded sources of nonlocality and randomness} 

\author[1]{Florian J. Curchod}
\author[1]{Markus Johansson}
\author[2]{Remigiusz Augusiak}
\author[3]{\hspace{2cm}Matty J. Hoban}
\author[1]{Peter Wittek}
\author[1,4]{Antonio Ac\'{i}n}
\affil[1]{ICFO-Institut de Ciencies Fotoniques, The Barcelona Institute of Science and Technology, 08860 Castelldefels (Barcelona), Spain}
\affil[2]{Center for Theoretical Physics, Polish Academy of Sciences, Al. Lotnik\'ow 32/46, 02-668 Warsaw, Poland}
\affil[3]{School of Informatics, University of Edinburgh, 10 Crichton Street, Edinburgh EH8 9AB, UK}
\affil[4]{ICREA--Instituci\'{o} Catalana de Recerca i Estudis
Avan\c{c}ats, E--08010 Barcelona, Spain}

\authorrunning{F. J. Curchod, M. Johansson, R. Augusiak, M. J. Hoban, P. Wittek and A. Ac\'{i}n} 

\begin{document}

\maketitle

\begin{abstract}
The outcomes of local measurements made on entangled systems can be \textit{certified} to be random provided that the generated statistics violate a Bell inequality. This way of producing randomness relies only on a minimal set of assumptions because it is independent of the internal functioning of the devices generating the random outcomes. In this context it is crucial to understand both qualitatively and quantitatively how the three fundamental quantities -- entanglement, non-locality and randomness -- relate to each other. To explore these relationships, we consider the case where repeated (non projective) measurements are made on the physical systems, each measurement being made on the post-measurement state of the previous measurement. In this work, we focus on the following questions: \textit{For systems in a given entangled state, how many nonlocal correlations in a sequence can we obtain by measuring them repeatedly? And from this generated sequence of non-local correlations, how many random numbers is it possible to certify?} In the standard scenario with a single measurement in the sequence, it is possible to generate non-local correlations between two distant observers only and the amount of random numbers is very limited. Here we show that we can overcome these limitations and obtain \textit{any} amount of certified random numbers from an entangled pair of qubit in a pure state by making sequences of measurements on it. Moreover, the state can be arbitrarily weakly entangled. In addition, this certification is achieved by near-maximal violation of a particular Bell inequality for each measurement in the sequence. We also present numerical results giving insight on the resistance to imperfections and on the importance of the strength of the measurements in our scheme. 
\end{abstract}

\section{Introduction}

Bell's theorem \cite{Bell1964} has shown that the predictions of quantum mechanics demonstrate non-locality. That is, they cannot be described by a theory in which there are objective properties of a system prior to measurement that satisfy the no-signalling principle (sometimes referred to as ``local realism"). Thus, if one requires the no-signaling principle to be satisfied at the operational level then the outcomes of measurements demonstrating non-locality must be unpredictable \cite{Bell1964,PR,Masanes2006}. This unpredictability, or randomness, is not the result of ignorance about the system preparation but is \textit{intrinsic} to the theory.

Although the connection between quantum non-locality (via Bell's theorem) and the existence of intrinsic randomness is well known \cite{Bell1964,PR,BellReview,Masanes2006} it was analyzed in a quantitative way only recently \cite{Randomness,Colbeck}. It was shown how to use non-locality (probability distributions that violate a Bell inequality) to \textit{certify} the unpredictability of the outcomes of certain physical processes. This was termed \textit{device-independent randomness certification}, because the certification only relies on the statistical properties of the outcomes and not on how they were produced. The development of information protocols exploiting this certified form of randomness, such as device-independent randomness expansion \cite{Randomness,Colbeck,Vazirani} and  amplification protocols~\cite{CR,Gallego}, followed.\\

Entanglement is a necessary resource for quantum non-locality, which in turn is required for randomness certification. It is thus crucial to understand qualitatively and quantitatively how these three fundamental quantities relate to one another. In our work, we focus on asking how many observers in a sequence can be nonlocally correlated and how much certified randomness can be generated from given entangled states as the resource when these are measured repeatedly. 
An important step to answer this question was recently made in \cite{Silva}, in which it was shown that nonlocality generated by a maximally entangled state can be shared between any number of distant observers, however, at the cost of exponentially diminishing the amount of nonlocality, as measured by the violation of the CHSH Bell inequality, between all the observers. Here we answer a significantly more demanding question that such correlations can be made arbitrarily close to extremal for each observer, a crucial property for randomness certification. In this particular sense we show that the nonlocality does not need to be diminished, as for each observer the generated correlations violate a particular Bell inequality (almost) maximally. 

For randomness certification, progress has been made for entangled states shared between two parties, Alice ($A$) and Bob ($B$), in the standard scenario where each party makes a single measurement on his share of the system and then discards it. An argument adapted from Ref. \cite{DAriano} shows that either of the two parties, A or B can certify at most $2\textrm{log}_2 d$ bits of randomness \cite{Acin2015}, where $d$ is the dimension of the local Hilbert space the state lives in, which in turn implies a bound of $4\textrm{log}_2 d$ bits when the two outputs are combined. This demonstrates a fundamental limitation for device-independent randomness certification in the standard scenario. The main goal of our work is to show that this limitation on the amount of certifiable random bits from one quantum state can be lifted. To do this we will consider the sequential scenario, where sequences of measurements can be applied to each local system. Our main result is to prove that an unbounded amount of random bits can be certified in this scenario.  \\

Imagine the following situation where, contrary to the device-independent approach that we follow in this article, one has perfect control over the functioning of the device generating randomness. An entangled state initially prepared in the Pauli-$Z$ basis, i.e., a $\sigma_z$ eigenstate $\ket{0}$ or $\ket{1}$, is measured in the Pauli-$X$, or $\sigma_x$ basis $\ket{\pm} = \frac{\ket{0}+\ket{1}}{\sqrt{2}}$. The outcome of this measurement is perfectly random and the post-measurement state is now one of the two eigenstates of the Pauli-$X$ basis $\ket{\pm}$. If the device now measures this new state in the original Pauli-$Z$ basis, the outcome of this new measurement is again random and one of the $\sigma_z$ eigenstates is obtained. A device alternating between measurements in those two orthogonal basis thus allows one to obtain any amount of random bits from a single state as input.

Of course, this way of generating randomness can never be trusted, as one can always design a classical device (with deterministic outcomes -- a local model) that has the same behavior as the device we described, i.e., their outputs are indistinguishable. To \textit{certify} randomness one needs the generation of non-local correlations, that can not be simulated with classical resources. But is it nevertheless possible to use this idea of measuring a state repeatedly, in a scheme exploiting non-locality, to obtain more random numbers and beat the bounds on randomness certification? Clearly, certifying more randomness by making sequences of measurements on the same state depends on whether one is able to produce sequences of non-local correlations between distant observers, as otherwise no additional randomness can be certified. One of the obstacles to this is that if local (projective) measurements are used to generate the non-local correlations, the entanglement in the state is destroyed. Then the post-measurement state is separable and thus cannot be further used to generate nonlocality or to certify randomness. A challenge is therefore to come up with measurements that do not destroy all the entanglement in the state but nevertheless generate non-local correlations. With such measurements the post-measurement state will still be a potential resource for the generation of more non-local correlations and certified randomness.

Bell tests with sequences of measurements have received less attention in the literature than the standard ones with a single measurement round despite the novel features in this scenario \cite{GallegoSeq}, as for example the phenomenon known as hidden nonlocality \cite{Popescu}. In our work we show that they prove useful in the task of randomness certification, which also provides another example~\cite{Acin2015} where general measurements can overcome limitations of projective ones. More precisely, we describe a scheme where any number $m$ of random bits are certified using a sequence of $n > m$ consecutive measurements on the same system. This work thus shows that the bound of $4\textrm{log}_2d$ random bits in the standard scenario can be overcome in the sequential scenario, where it is impossible to establish any bound. The unbounded randomness is certified by a near-maximal violation of a particular Bell inequality for each measurement in the sequence. Moreover, for any finite amount of certified randomness, our scheme has a finite (yet very small) noise robustness. Our results show that entangled systems are an unbounded resource for the generation of non-local correlations and of certified random numbers.\\

This paper is an extended version of \cite{curchod2017unbounded}, where the main results are already included. The rest of the paper is organized as follows. In section \ref{SequScen}, we describe the sequential scenario that allows for multiple measurements on the same state. In section \ref{Guess}, we generalize the concept of guessing probabilities -- that allow to certify upper bounds on the predictive power of an adversary trying to guess the random numbers -- to the sequential scenario and obtain new results on their continuity properties. In section \ref{Nondestr} we introduce the main ingredients we will use in our scheme, in particular we introduce a family of measurements on two qubit states that allow us to retain some entanglement in the post-measurement states. In section \ref{Scheme} we describe our scheme that allows for the generation of nonlocal correlations between any number of distant observers and any amount of certified random numbers. In section \ref{Numerics} we present numerical results on the relation between the amount of violation of the family of inequalities introduced in \cite{Acin2012} and the amount of randomness that can be certified from it. In section \ref{Weak} we obtain numerical results to understand the relation between the strength of the measurement and the amount of randomness that can be certified from it. We conclude in section \ref{Conclusions} with additional remarks and potential future work.

\section{The sequential measurements scenario}\label{SequScen} Before presenting our results, let us introduce the scenario we work in. We carry over many of the features from the standard scenario except now we allow party $B$ to make multiple measurements in a sequence on his share of the state. One can visualize this as in Fig. \ref{Fig:scen} where $B$ is split up into several $B$s, each one corresponding to a measurement made on the state and labeled by $B_i$, $i \in \{1,2,..,n\}$, where $n$ is the total number of measurements made in the sequence. Each $B_i$ makes one measurement and the post-measurement state is sent to $B_{i+1}$. We organize the Bobs such that $B_i$ is doing his measurement \textit{before} $B_j$ for $i<j$. Thus in principle $B_j$ can receive the information about the inputs and outputs of previous measurements $B_i$ for all $i<j$.  

\begin{figure}[h!]
\scalebox{0.24}{\includegraphics{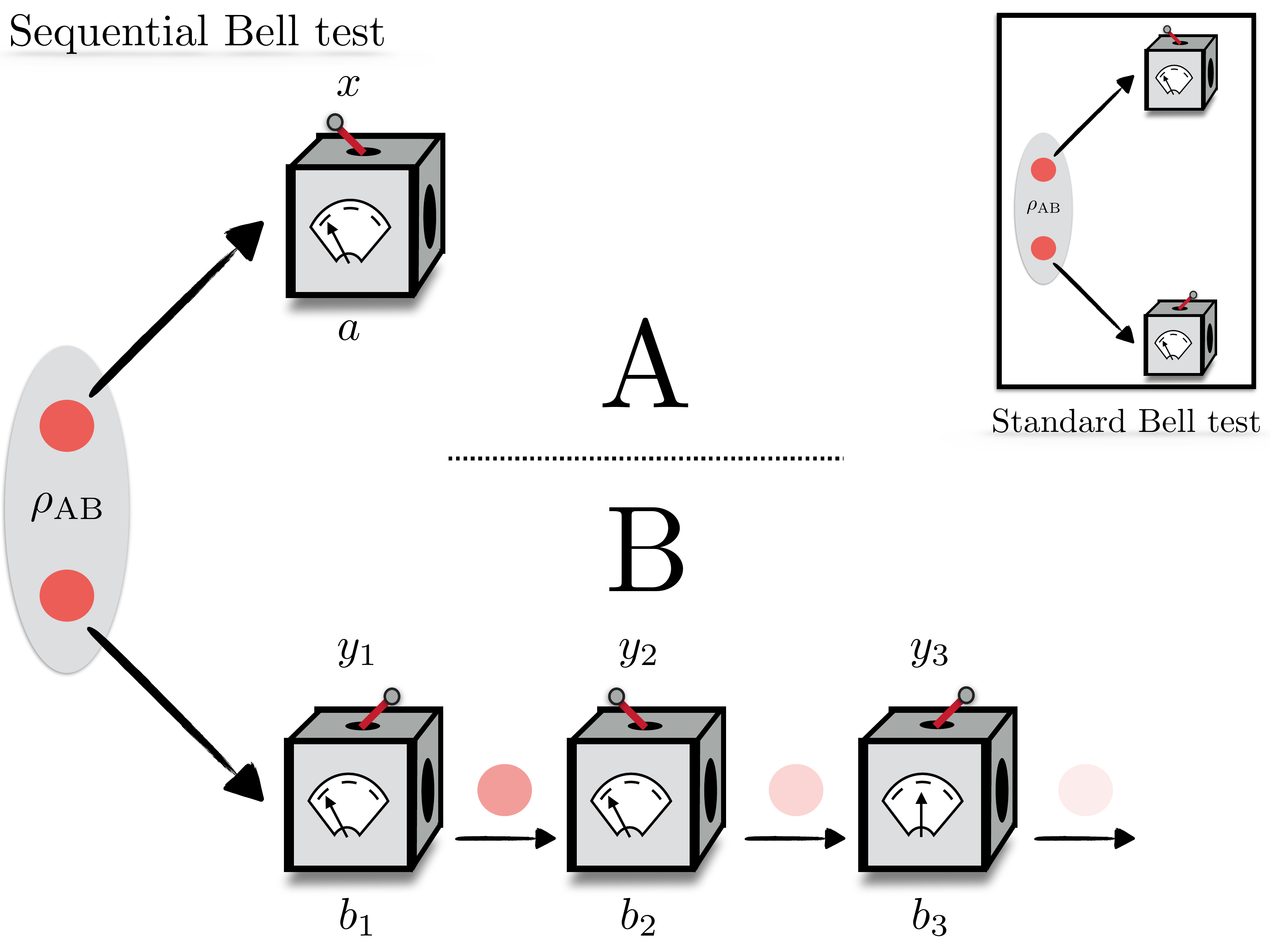}}
\caption{\label{Fig:scen}
The standard scenario, where parties $A$ and $B$ make a single quantum measurement on their share of the state and discard it versus the sequential scenario where the second party $B$ makes multiple measurements on his share.}
\end{figure}

\section{Randomness certification: from the standard to the sequential scenario}\label{Guess}
To quantify the randomness produced in the setup, we put the above scenario in the setting of \textit{non-local guessing games} (e.g. Refs. \cite{Acin2012, Nieto, DeLaTorre, Acin2015}). Let us consider an additional adversary Eve ($E$) who is in possession of a quantum system potentially correlated to the one of $A$ and $B$. The global state is denoted $\rho_{ABE}$. We assume that at each round of the experiment, $E$ is the one preparing the state $\rho_{ABE}$ and distributes $\rho_{AB} = \mathrm{Tr}_E \rho_{ABE}$ to $A$ and $B$. This state will be used to make the measurements in the sequence and the aim of $E$ is to try to guess $B$'s outcomes by using measurements on her share of the state $\rho_{ABE}$. The parties $A$ and $B_i$s, having no knowledge about the state or the real measurements made on it,  see their respective devices as black boxes that receive some classical input $x \in \{0,1\}$ and $y_i \in \{0,1\}$, $y_1,y_2,..,y_n \equiv \vec{y}$, respectively, and that
 generate a classical output $a \in \{\pm 1\}$ and $b_i \in \{\pm 1\}$, $(b_1,b_2,..,b_n) \equiv \vec{b}$, respectively (see Fig. \ref{Fig:scen}). They generate statistics from multiple runs of the experiment to obtain the observed probability distribution $P_{\textrm{obs}}$ with elements $p_{\textrm{obs}}(a,\vec{b}|x,\vec{y})$. This distribution $P_{\textrm{obs}}$ lives inside the set of quantum correlations $\mathcal{Q}_n$ obtained from measurements on quantum states in a sequence as we described. This set is convex and thus can be described in terms of its extreme points, denoted $P_{\textrm{ext}}$, and any $P_{\textrm{obs}}$ can be written as  $P_{\textrm{obs}}=\sum\limits_{\textrm{ext}}q_\textrm{ext}P_{\textrm{ext}}$, where $\sum\limits_{\textrm{ext}}q_\textrm{ext}=1$ and every $q_\textrm{ext}\geq 0$.

From studying the outcome statistics \textit{only} we can bound $E$'s predictive power by allowing her to have complete knowledge of how $P_{\textrm{obs}}$ is decomposed into extreme points, i.e., she knows the probability distribution $q_\textrm{ext}$ over extreme points $P_{\textrm{ext}}$. This predictive power is quantified via the \textit{device-independent guessing probability} (DIGP)~\cite{Acin2012} where we fix the particular input string $y_1^0,y_2^0,..,y_n^0 \equiv \vec{y}^0$ for which $E$ has to guess the outputs $\vec{b}$. The DIGP, denoted by $G(\vec{y}^0,P_{\textrm{obs}})$, is then calculated as the optimal solution to the following optimization problem \cite{DeLaTorre,Nieto}:
\begin{align}
& G(\vec{y}^0,P_{\textrm{obs}})= \max_{\{q_\textrm{ext},P_{\textrm{ext}}\} } \sum_{\textrm{ext}} q_\textrm{ext} \max_{\vec{b}}p_{\textrm{ext}}(\vec{b}|\vec{y}^0) \nonumber \\
 &\text{subject to:} \nonumber \\
 &p_{\textrm{ext}}(\vec{b}|\vec{y}^0)=\sum_{a}p_{\textrm{ext}}(a,\vec{b}|x,\vec{y}^0),\hspace{1cm}\forall x\\
&P_{\textrm{obs}}=\sum_{\textrm{ext}}q_\textrm{ext}P_{\textrm{ext}},\hspace{2.3cm}P_{\textrm{ext}}\in \mathcal{Q}_n. \label{sumupto}
\end{align}
The operational meaning of this quantity is clear: Eve has a complete description of the observed correlations in terms of extreme points. She then guesses the most probable outcome for each extreme point. The standard scenario with a single measurement round can also be represented in this formalism by simply considering that $\vec{b} = b$ and $\vec{y}^{(0)} = y^{(0)}$. To quantify the amount of bits of randomness that is certified, we use the \textit{min entropy} $H(\vec{y}^0,P_{\textrm{obs}})=-\log_{2}G(\vec{y}^0,P_{\textrm{obs}})$ which returns $m$ bits of randomness if $G(\vec{y}^0,P_{\textrm{obs}})=2^{-m}$. The amount of bits of randomness quantified in this way is the figure of merit in this work and our goal is to obtain as many bits as possible from a single system.\\

We will now derive some general properties of the guessing probability \eqref{sumupto} in the form of theorems \ref{theorem1} and \ref{SequTheo}. Let us stress here that these results are not limited to the guessing probability used in this work but are general properties of guessing probabilities. A more detailed discussion and an introduction to the topic of guessing probabilities and their use in randomness certification can be found in the appendices, as well as the proofs of the theorems that we discuss here.

For a single measurement on each system (i.e. a sequence of $n = 1$ measurement), which corresponds to the standard Bell scenario and $\mathcal{Q} \equiv \mathcal{Q}_1$ the set of quantum correlations for a single measurement on each subsystem we have that:
\begin{proposition}\label{prop1}
The function $G(y^0,P_{\textrm{obs}})$ on the set of quantum distributions $\mathcal{Q}$ is continuous in the interior of $\mathcal{Q}$.
\end{proposition}
\begin{proposition}\label{prop2}
The function $G(y^0,P_{\textrm{obs}})$ is continuous in any extremal point of $\mathcal{Q}$.
\end{proposition}
The proofs of these two propositions are based mostly on general properties of concave functions \cite{rock} and of concave roof extensions in particular \cite{buci}, and can be found in section \ref{APP2:ContinuityProof} of the appendices. In other words the guessing probability for a single measurement is continuous everywhere except possibly on some points that lie on the surface of the quantum set but that are not extremal. An example of this can be obtained from the measurements described in \cite{Randomness} for a state with arbitrarily little entanglement. The joint conditional probability distribution (introduced below, see \eqref{Eq:Ibeta2}) corresponding to those measurements made on such a state has $G(y^0,P_{obs})=1/2$ and is at the same time arbitrarily close to a joint conditional probability distribution corresponding to measurements on a product state with $G(y^0,P_{obs})=1$, i.e., a local point. The key is that this local point is not extremal, it lies somewhere on the surface of the local (and quantum) set but can be decomposed into other extremal (local) points, i.e. is not a vertex of the local polytope. Discontinuities of $G(y^0,P_{\textrm{obs}})$ can thus appear only at the boundary between extremal points and non-extremal points lying on the surface of the set, and in the rest of the set it is continuous.\\

In general -- and in particular in our work -- the optimization problem \eqref{sumupto} can be relaxed to an optimization where instead of insisting on $P_{\textrm{obs}}=\sum\limits_{\textrm{ext}}q_\textrm{ext}P_{\textrm{ext}}$ \eqref{sumupto}, one only imposes that the observed statistics $P_{\textrm{obs}}$ give a particular Bell inequality violation \cite{Randomness}. The optimal solution to this new problem is an upper bound to the optimal solution of \eqref{sumupto}. Crucially, this relaxation often gives non trivial bounds as shown in our case for example. From now on, every time we refer to a guessing probability we refer to this relaxation of the problem to a particular Bell inequality violation.

Now we consider a Bell expression $I$ with its maximal value $t_{\textrm{max}}$ on the quantum set $\mathcal{Q}$. We define the hyperplane $H_t$ to contain the elements of $\mathcal{Q}$ for which the value of I is $t \leq t_{\textrm{max}}$ and further we define the restriction $G(y^0,P_{\textrm{obs}})_t$ of $G(y^0,P_{\textrm{obs}})$ to the intersection of $H_{t}$ with $\mathcal{{{Q}}}$ and let $\max G(y^0,P_{\textrm{obs}})_t$ be the maximum of the guessing probability on this intersection. From Propositions ($1$) and ($2$) we can show that:
\begin{theorem}\label{theorem1}
If the intersection of $H_{t_\textrm{max}}$ with $\mathcal{Q}$ is a single (thus extremal) point, there exists a $t_{\textrm{c}} < t_{\textrm{max}}$ such that $G(y^0,P_{\textrm{obs}})_t$ is a continuous function of $t$ for $t_{\textrm{c}} \leq t \leq t_{\textrm{max}}$
\end{theorem}
The proof of this theorem can be found in section \ref{APP6:ContAtPextUnique} of the appendices. In the other case, if the intersection of $H_{t_\textrm{max}}$ with $\mathcal{Q}$ has more than one point, it also contains a set of non-extremal points of $\mathcal{Q}$ and therefore a discontinuity of $G(y^0,P_{\textrm{obs}})_{t}$ at $t_{\textrm{max}}$ can not be ruled out by theorem \eqref{theorem1}. In other words, if the violation of a particular Bell inequality $I$ is achieved by a unique quantum point (as for example the following \eqref{Eq:Ibeta}), the guessing probability close to that point is continuous.\\

Until now, we have considered the continuity properties of the guessing probability in the standard scenario with a single measurement in the sequence. Now we would like to extend those results to the guessing probability in the sequential measurement scenario with $n \geq 2$ measurements being made on the subsystems. Remember that we split party $B$  into many $B_i$, so that party $B_i$ makes the $i$th measurement on the system. The measurement setting of $B_i$ is $y_i$ and its outcome $b_i$ (see Fig. 1). In our work, we will always take $y_i \in \{0,1\}$ and $b_i \in \{0,1\}$, but the following results can be generalized to any number of inputs and outcomes (they may even be different for each measurement in the sequence).

Now consider the joint conditional probability distributions $P^i_{\textrm{obs}}(a,b_i|x,y_1,...,y_i,b_1,...,b_{i-1})$ between $A$ and each $B_i$, that is the joint conditional probability distribution between $A$ and $B_i$ conditioned on what happened before the $i$th measurement, namely the input choices $y_1,...,y_{i-1}$ and the outcomes $b_1,...,b_{i-1}$ that were obtained \textit{before} measurement $i$. There are $n$ of those joint conditional probability distributions living in $\mathcal{Q}$ that can be obtained directly from the whole probability distribution for the sequence $P_{\textrm{obs}}(a\vec{b}|x\vec{y})$ living in $\mathcal{Q}_n$. Now suppose that we play, for each distribution $P^i_{\textrm{obs}}(a,b_i|x,y_1,...,y_i,b_1,...,b_{i-1})$, a Bell game $I_i$ such that $I_i(P_i(a,b_i|x,y_1,...,y_i,b_1,...,b_{i-1})) = t_i \leq t_i^{\textrm{max}}$, where $t_i^{\textrm{max}}$ is the maximum of $I_i$ over the set $\mathcal{Q}$.


\begin{theorem}\label{SequTheo}
Suppose that each joint conditional probability distribution $P^i_{\textrm{obs}}(a,b_i|x,y_1,...,y_i,b_1,...,b_{i-1})$ between $A$ and $B_i$ in the sequence is such that $I_i(P_i(a,b_i|x,y_1,...,y_i,b_1,...,b_{i-1})) = t_i$ and consider the limit where each $t_i \rightarrow t_i^{\textrm{max}}$. Suppose also that for each $i$, $G_i(y_i^0,P^i_{\textrm{obs}}(a,b_i|x,y_1,...,y_i,b_1,...,b_{i-1}))$ attains its smallest possible value at $t_i=t_i^{\textrm{max}}$. Then if the maximal value $t_i^{\textrm{max}}$ of each $I_i$ is achieved in a unique quantum point in $\mathcal{Q}$:
\begin{equation}
G(\vec{y}^0,P_{\textrm{obs}}(a\vec{b}|x\vec{y})) \rightarrow \prod\limits_{i=1}^{n} G_i(y_i^0,P^i_{\textrm{obs}}(a,b_i|x,y_1,...,y_i,b_1,...,b_{i-1})) 
\end{equation}
\end{theorem}
where $G_i(y_i^0,P^i_{\textrm{obs}}(a,b_i|x,y_1,...,y_i,b_1,...,b_{i-1}))$ is the (non sequential) relaxed guessing probability \eqref{sumupto} of an adversary $E$ trying to guess outcome $b_i$ for input $y_i^0$ from the observed joint probability distribution $P^i_{\textrm{obs}}(a,b_i|x,y_1,...,y_i,b_1,...,b_{i-1}))$. The proof of this theorem can be found in appendices \ref{APP3:PGforSequ} and \ref{APP4:ArbitrarySequ}. In other words, if each measurement in the sequence taken separately -- thus not seen as in a sequence -- leads to correlations close enough to the unique maximal violation of inequality $I_i$ between $A$ and $B_i$ only, and if this maximal violation corresponds to the minimal possible guessing probability for $b_i$, then the guessing probability for the whole sequence tends to the product of the individual guessing probabilities of the outcomes $b_i$.\\

Before presenting our results, it is worth explaining why the causal constraints imposed by the sequential scenario make it stronger than standard Bell tests with one measurement in the sequence. At first sight, one could be tempted to group all the measurements in the sequence into a single box receiving an input string $\vec{y}_n$ to output another string $\vec{b}_n$, as in a standard Bell test. However, in general a sequence of measurements can not be represented as a single measurement. To understand this, note that in the sequential scenario the outcome $b_i$ can depend only on variables produced in its past, namely the input choices $y_1,y_2,...,y_i$ and the outcomes $b_1,b_2,...,b_{i-1}$ that were \textit{previously} obtained. However, in the single measurement scenario, the measurement box receives all inputs and produces all outputs at once. In particular, outcome $b_i$ can now be a function of input choices $y_{j>i}$ and outcomes $b_{j>i}$ that are produced in the \textit{future}. That is, such a big box may violate the physical constraints coming from the sequential arrangement and the assumption that signaling from the future to the past is impossible. These additional causality constraints further limit Eve's predictability with respect to a standard Bell test and are responsible of the unbounded amount of certified randomness.

\section{Making non-destructive measurements on qubit states}\label{Nondestr} 

Alice and Bob share the pure two-qubit state
\begin{equation}\label{Eq:qubits}
	\ket{\psi(\theta)} = \cos(\theta)\ket{00}+\sin(\theta)\ket{11}
\end{equation}
that for all $\theta \in ]0,\pi/2[$ is entangled. In Ref. \cite{Acin2012}, a family of Bell inequalities was introduced: 
\begin{equation}\label{Eq:Ibeta}
I_{\theta} = \beta\langle \mathbb{B}_0 \rangle + \langle \mathbb{A}_0\mathbb{B}_0 \rangle + \langle \mathbb{A}_1\mathbb{B}_0 \rangle\\ + \langle \mathbb{A}_0\mathbb{B}_1 \rangle - \langle \mathbb{A}_1\mathbb{B}_1 \rangle
\end{equation}
where $\beta=2\cos(2 \theta)/[1+\sin^2(2 \theta)]^{1/2}$,  $\langle \mathbb{B}_y \rangle = p(b=+1|y)-p(b=-1|y)$ and $\langle \mathbb{A}_x \mathbb{B}_y \rangle = p(a=b|xy)-p(a \neq b|xy)$ for $x$, $y\in\{0,1\}$.
This family of inequalities has the following two useful properties: first, its maximal quantum violation, $I_{\theta}^{\max} = 2\sqrt{2}\sqrt{1+\beta^2/4}$, is obtained by measuring the state \eqref{Eq:qubits} with measurements:
\begin{equation}\label{Eq:Ibeta2}
\begin{split}
\mathbb{A}_0 = \cos\mu\,\sigma_z + \sin\mu\,\sigma_x, \hspace{1cm} \mathbb{B}_0 = \sigma_z, \\
\mathbb{A}_1 = \cos\mu\,\sigma_z - \sin\mu\,\sigma_x, \hspace{1cm} \mathbb{B}_1 = \sigma_x,
\end{split}
\end{equation}
where $\tan\mu = \sin(2\theta)$.
Second, when maximally violated, the inequality certifies one bit of local randomness on Bob's side for his second measurement choice $y^0=1$: $G(y^0=1,P^{\textrm{max}}_{\textrm{obs}})=1/2$ \cite{Acin2012}. These observations are possible because the maximal violation is \textit{uniquely} achieved by the probability distribution $P^{\textrm{max}}_{\textrm{obs}}$ that arises from the previously-described state and measurements \eqref{Eq:qubits} and \eqref{Eq:Ibeta2}. Therefore, for the maximal violation, $P_{\textrm{obs}}^{\textrm{max}}=P_{\textrm{ext}}$ in \eqref{sumupto} and the guessing probability for input choice $y^0=1$ is equal to $1/2$.

However, in general we may not get correlations that maximally violate our Bell inequality but give a violation that is only close to maximal. In section \ref{Guess} we have shown how to make conclusions about the guessing probability for non-maximal violations. In particular, we showed that for \textit{any} Bell inequality with a unique point of maximal violation, the guessing probability is a continuous function of the value of the inequality close to the maximal violation. This implies in the particular case we are studying that:
\begin{equation}\label{Eq:cont}
	I_{\theta} \rightarrow I_{\theta}^{max} \hspace{0.5cm} \Rightarrow \hspace{0.5cm} G(y^0=1,P_{\textrm{obs}}) \rightarrow \frac{1}{2}.
\end{equation}
In section \ref{Numerics}, we also provide a numerical upper bound on the guessing probability $G(y^0=1,P_{obs})$ by a concave function of the value of $I_{\theta}$. 

Bell inequalities~\eqref{Eq:Ibeta} are the first main ingredient in our sequential construction below. The second one is the use of general, non-projective measurements. Indeed, if $B_1$ performs a projective measurement on the shared entangled state, the resulting post-measurement state, now shared between Alice and $B_{2}$, is separable and thus useless for randomness production. Consequently, one needs to consider non-projective measurements to retain some entanglement in the system for the subsequent measurements. For this purpose, let us introduce the following two-outcome quantum measurement (written in the formalism of Kraus operators):
\begin{equation}\label{MeasKraus}
M_{\pm1}(\xi)=\cos\xi|\pm\rangle\!\langle\pm|+\sin\xi|\mp\rangle\!\langle\mp|
\end{equation}
corresponding to the two outcomes $\{\pm1\}$. This measurement $\hat{\sigma}_x(\xi)\equiv\{M_{+1}^{\dagger}M_{+1},M_{-1}^{\dagger}M_{-1}\}$
can be understood as a generalization of the projective measurement $\sigma_x$. It varies from being projective (for $\xi = 0$) to being non-interacting (for $\xi = \pi/4$). One can verify that measuring an entangled state \eqref{Eq:qubits} for $\xi \in ]0,\pi/4]$ (non-projective measurement) the post-measurement state still retains some entanglement, irrespectively of the outcome. Therefore, by tuning the parameter $\xi$ we are able to vary the destruction of the entanglement of the state at the gain of extracting information from it (cf. Ref. \cite{Silva}): the closer to being a projective measurement, the lower the entanglement in the post-measurement state, but the bigger the violation of the initial Bell inequality.

\section{A scheme for an unbounded amount\\of  nonlocal correlations and certified random numbers}\label{Scheme}
We now combine the previous observations to demonstrate our main result. First, let us recall that, as shown in \cite{Acin2012}, one can obtain one bit of randomness from any pure entangled two qubit state, irrespective of the amount of entanglement in it. Moreover, one can verify that approximately one random bit can be certified if the measurements are close to the ones in Eq. \eqref{Eq:Ibeta2} (in the sense that $\hat{\sigma}_x(\xi)$  is close to a measurement of $\sigma_x$ for $\mathbb{B}_1$ in Eq. \eqref{Eq:Ibeta2}) since $I_{\theta}$ is then close to $I_{\theta}^{\max}$ in Eq. \eqref{Eq:cont}. Second, the measurement in Eq. \eqref{MeasKraus} is only close to projective for $\xi$ close to zero and leaves entanglement in the post-measurement state between Alice and Bob which is thus still useful for randomness certification.
By repeated use of these two properties we can certify the production of an unbounded amount of random bits from a single pair of entangled qubits. We now formally describe this process in which Alice makes a single measurement on her share of the state, whereas Bob makes a sequence of $n$ measurements on his.

Each $B_i$ chooses between measurements of $\sigma_z$ and $\hat{\sigma}_x(\xi_i)$ \eqref{MeasKraus} for inputs $y_{i}=0$ and $y_{i}=1$, respectively, with outcomes $b_i \in \{\pm 1\}$. The parameter $\xi_i$ is fixed before the beginning of the experiment. The initial entangled state shared between Alice and Bob, before $B_1$'s measurement, is $\ket{\psi^{(1)}(\theta_1)}$ (see Eq. \eqref{Eq:qubits} with $\theta = \theta_1$). If the first non-projective measurement of the operator $\hat{\sigma}_x(\xi_1)$ is made by $B_1$ on the initial state $\ket{\psi^{(1)}(\theta_1)}$, the post-measurement state is of the form
\begin{equation}\label{Eq:diagpostsele}
\ket{\psi^{(2)}_{b_1}(\theta_1,\xi_1)}=U_A^{b_1}(\theta_1,\xi_1) \otimes V_B^{b_1}\left(\theta_1,\xi_1)(c\ket{00}+s\ket{11}\right),
\end{equation}
where $c=\cos(\theta_{b_1}(\theta_1,\xi_1))$ and $s=\sin(\theta_{b_1}(\theta_1,\xi_1))$ and the two unitaries, $U_A^{b_1}(\theta_1,\xi_1)$ and $V_B^{b_1}(\theta_1,\xi_1)$, and angle $\theta_{b_1}(\theta_1,\xi_i) \in ]0,\pi/4]$ depend on the first outcome $b_{1}$ and the angles $\theta_1$ and $\xi_1$.

After his measurement, $B_1$ applies the unitary $(V_B^{b_1})^\dagger$, conditioned on his outcome $b_1$, on the post-measurement state going to $B_{2}$. This allows $B_2$ to use the  same two measurements $\hat{\sigma}(\xi_2)$ and $\sigma_z$ independently of the outcome $b_1$ since the unitary $(V_B^{b_1})$ is canceled in \eqref{Eq:diagpostsele}. This last procedure will be applied by each $B_i$ after his measurement, before sending the post-measurement state to the next $B_{i+1}$. If the system passed through \textit{only} the non-projective measurements, the state received by $B_i$ can be one of $2^{i-1}$ potential states, depending on all of the previous $B_j$'s ($j<i$) outcomes (one for each combination $\vec{b}_{i-1} \equiv (b_1,b_2,..,b_{i-1})$ of outcomes obtained by the previous $B_j$, these can be computed \textit{before} the beginning of the experiment). Any of these states can be written as:
\begin{equation}\label{Eq:statei}
\ket{\psi^{(i)}_{\vec{b}_{i-1}}} = U_A^{\vec{b}_{i-1}} \otimes \mathbbm{1}_B \left[\cos(\theta_{\vec{b}_{i-1}})\ket{00}+\sin(\theta_{\vec{b}_{i-1}})\ket{11}\right],
\end{equation}
where the angles $\theta_{\vec{b}_{i-1}}$ and the matrix $U_A^{\vec{b}_{i-1}}$ both depend on the outcomes $\vec{b}_{i-1}$, on the initial angle $\theta_1$ and the angles $\xi_{j}$ of the previous $B_j$'s with $j<i$. In the notation, we will always omit the dependence on the angles $\theta_1$ and $\xi_1,\xi_2,..,\xi_{j}$ since these are fixed \textit{before} the beginning of the experiment. For each of these different potential states with angle $\theta_{\vec{b}_{i-1}}$, Alice adds two measurements to her input choices, where for $k\in\{0,1\}$, these are measurements of the observables $\mathbb{A}_{k}^{\vec{b}_{i-1}}$ which are defined as
\begin{equation}\label{Eq:Alice2Meas}
U_A^{\vec{b}_{i-1}}\left[\cos(\mu_{\vec{b}_{i-1}})\sigma_z+(-1)^{k}\sin(\mu_{\vec{b}_{i-1}})\sigma_x\right](U_A^{\vec{b}_{i-1}})^{\dagger},
\end{equation}
where $\tan(\mu_{\vec{b}_{i-1}})=\sin(2\theta_{\vec{b}_{i-1}})$, depending on the specific state $\ket{\psi^{(i)}_{\vec{b}_{i-1}}}$ \eqref{Eq:statei}.

We are now ready to describe how the scheme certifies randomness.
The measurement operator $\hat{\sigma}_x(\xi_i)$ can be made arbitrarily close to $\sigma_x$ by choosing $\xi_{i}$ sufficiently small. This brings the outcome statistics for measurements $\hat{\sigma}_x(\xi_i), \sigma_z$ on Bob's side and $\mathbb{A}_0^{\vec{b}_{i-1}},\mathbb{A}_1^{\vec{b}_{i-1}}$ on Alice's side on the state in Eq. \eqref{Eq:statei}, arbitrarily close to the statistics for the measurements in Eq. \eqref{Eq:Ibeta2} and a state of the form in Eq. \eqref{Eq:qubits}, for $\theta=\theta_{\vec{b}_{i-1}}$. Therefore, the inequality $I_{\theta_{\vec{b}_{i-1}}}$ for Alice and $B_i$ as defined in \eqref{Eq:Ibeta} can be made arbitrarily close to its maximal violation. This in turn guarantees that the guessing probability, $G(y^0_i=1,P_{obs})$ can be made arbitrarily close to $1/2$.
Note that this guessing probability does not only describe the instances when Alice chooses the measurements $\mathbb{A}_j^{\vec{b}_{i-1}}$. Since Eve does not know Alice's measurement choices in advance she cannot use a strategy that gives higher predictive power for the instances when Alice chooses other measurements. Finally, by making $G(y^0_i=1,P_{obs})$ sufficiently close to $1/2$ for each $i$ (by choosing each $\xi_i$ sufficiently close to $0$) the DIGP $G(y^0_1,y^0_2,..,y_n^0,P_{obs})$ can, by continuity, be made arbitrarily close to $2^{-n}$ (see theorem \ref{SequTheo} of section \ref{Guess}.)

At the end, Bob can produce $m$ random bits by a suitably chosen sequence $\hat{\sigma}_x(\xi_i)$, $i \in \{1,2,..,n\}$, of $n>m$ measurements. The certification only requires that each $B_i$ occasionally chooses the projective measurement $\sigma_z$ so that the whole statistics can be obtained. Note that Bob can choose $\sigma_z$ with probability $\gamma_i$ and $\hat{\sigma}_x(\xi_i)$ with probability $1-\gamma_i$ for $\gamma_i$ as close to zero as he wants. 
Finally, note that the value of \textit{each} inequality $I_{\theta_{\vec{b}_{i-1}}}$ between each $B_i$ and $A$ can be made as close as wanted to the maximal value $I_{\theta_{\vec{b}_{i-1}}}^{\textrm{max}}$. Therefore, we can certify randomness for each measurement $B_i$ in the sequence at the expense of increasing the number of measurements that Alice chooses from.  

This protocol can also be used to certify any finite amount of randomness with some small but strictly non-zero noise robustness. Indeed, assume the goal is to certify $m$ random bits. One can then run the protocol for $m'>m$ bits. By continuity, when adding a small but finite amount of noise the protocol will certify $m$ random bits. Of course, the noise robustness tends to zero with the number of certified random bits. However, we expect this to be the case for any protocol. This conjecture is based on the following argument: each measurement of a particle of finite dimension can produce only a finite amount of randomness. Thus, to get unbounded randomness, an infinite number of measurements are needed. Moreover, a measurement that is very close to non-interacting is unlikely to produce nonlocal correlations and is thus useless to certify randomness. It therefore appears quite likely that, in the infinite limit, any sequence of local measurements that are useful for randomness certification will destroy all the entanglement in the state, so that the resulting noise resistance tends to zero. We therefore expect that, while quantitative improvements over our protocol in terms of noise robustness can be expected, from a qualitative point of view it goes as far as possible.

\section{Numerical bounds on the amount of violation of the family of Bell inequalities of \cite {Acin2012} and the certified randomness}\label{Numerics}

Let us now explain some numerical results that should provide some quantitative intuition on the relation between the amount of violation of the family of inequalities \eqref{Eq:Ibeta} and the amount of random bits certified by this violation. This allows one to evaluate how close the value $I_{\theta}$ of the inequalities \eqref{Eq:Ibeta} should be to the maximal one $I_{\theta}^{max}$ in order to certify close to one perfect random bit from the statistics for one measurement $n = 1$.

Let us consider the following two-parameter class of Bell inequalities:
%
\begin{equation}\label{Ialphabeta}
I_{\alpha,\beta} := \beta\langle \mathbb{B}_0 \rangle + \alpha( \langle \mathbb{A}_0\mathbb{B}_0 \rangle + \langle \mathbb{A}_1\mathbb{B}_0 \rangle) + \langle \mathbb{A}_0\mathbb{B}_1 \rangle - \langle \mathbb{A}_1\mathbb{B}_1 \rangle\leq \beta+
2\alpha
\end{equation}
where $\alpha\geq 1$ and $\beta\geq 0$ such that $\alpha\beta<2$. For $\alpha=1$ the above class reproduces the family of Bell inequalities \eqref{Eq:Ibeta}
with $\beta = 2\cos(2\theta)/[1+\sin^2(2\theta)]^{1/2}$. In \cite{Acin2012} it was proved that the maximal quantum value $I_{\alpha,\beta}^{\textrm{max}}$ for these inequalities is given by:
\begin{equation}\label{maxQ}
I_{\alpha,\beta}^{\textrm{max}} = \sqrt{(1+\alpha^2)(4+\beta^2)}.
\end{equation}

Now, we conjecture that the following inequality is satisfied by 
$I_{\alpha\beta}$:
\begin{equation}\label{nierownosc}
I_{\alpha,\beta}^2+(2-\alpha\beta)^2\langle \mathbb{B}_1\rangle^2\leq (1+\alpha^2)(4+\beta^2).
\end{equation}
We have numerically evaluated this inequality for various values of $\alpha$ and $\beta$ by maximizing its left-hand side over general one-qubit measurements $\mathbb{A}_i=\vec{m}_i\cdot\vec{\sigma}$ and $\mathbb{B}_i=\vec{n}_i\cdot\vec{\sigma}$ with $\vec{m}_i,\vec{n}_i\in\mathbb{R}^3$ such that $|\vec{m}_i|=|\vec{n}_i|=1$ for $i=0,1$, and two-qubit pure entangled states that can always be written as
\begin{equation}
\ket{\psi}=\cos t \ket{00}+\sin t\ket{11}
\end{equation}
with $t\in[0,\pi/2]$ now being independent of $\beta$. The obtained values were always smaller than or equal to the right-hand side of (\ref{nierownosc}). Notice that in the case of Bell scenarios with two dichotomic measurements one can always optimize expression like the above one over qubit measurements and states (see e.g. Ref. \cite{Acin2012}).

From \eqref{nierownosc}, it is easy to obtain an upper bound on the expectation value:
\begin{equation}\label{numupper}
|\langle \mathbb{B}_1 \rangle | \leq \frac{\sqrt{(1+\alpha^2)(4+\beta^2)-I_{\alpha,\beta}^2}}{2-\alpha\beta} = \frac{\sqrt{(I_{\alpha,\beta}^{\textrm{max}})^2-I_{\alpha,\beta}^2}}{2-\alpha\beta},
\end{equation}
which, due to the fact that the right-hand side of the above is a concave function in $I_{\alpha,\beta}$, implies an upper bound on the guessing probability: 
\begin{equation}\label{Pnumupper}
G(y^{0}=1,P_{obs}) \leq \frac{1}{2} + \frac{\sqrt{(I_{\alpha,\beta}^{\textrm{max}})^2-I_{\alpha,\beta}^2}}{2(2-\alpha\beta)} \equiv f(I_{\alpha\beta}).
\end{equation}
In the particular case of maximal violation of the inequality $I_{\alpha\beta}$ \eqref{Ialphabeta} -- which saturates inequality \eqref{nierownosc}, this bound implies that the outcome of the first Bob's measurement is completely unpredictable, $G(y^{0}=1,P_{obs})=1/2$. 
Our numerical bound is thus tight at the maximal quantum violation of the inequality, but also when $I_{\alpha\beta}$ attains its classical value $2\alpha + \beta$, for which $G(y^0 = 1, P_{\textrm{obs}}) = 1$. In general, however, the bound \eqref{Pnumupper} is not tight. Still, it provides a good bound on the guessing probability in terms of the amount of violation of $I_{\alpha\beta}$ \eqref{Ialphabeta} and thus also of the family of inequalities $I_{\theta}$ \eqref{Eq:Ibeta} we were using in our scheme.\\

\begin{figure}[h!]
\scalebox{0.3}{\includegraphics{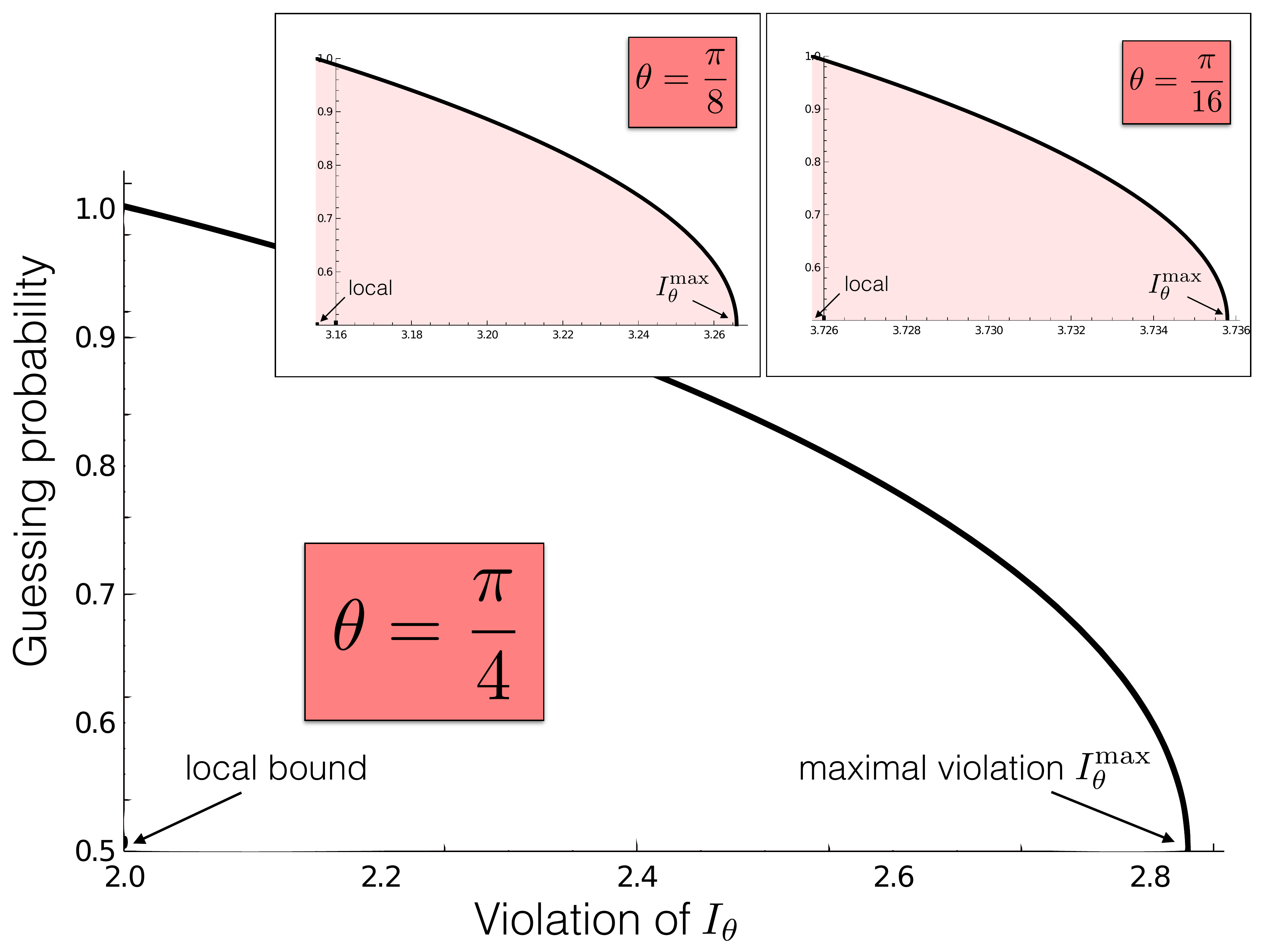}}
\caption{\label{Fig:graphs}
Our numerical upper bounds on the guessing probability in function of the violation of $I_{\theta}$ for $\theta = \frac{\pi}{4},\frac{\pi}{8},\frac{\pi}{16}$, where $I_{\theta = \frac{\pi}{4}} = $ CHSH. One can see that these are tight both at the maximal violation of the inequality and at its local bound.}
\end{figure}

%
For example, one can insert the maximal quantum value $I_\theta^{\textrm{max}}$ \eqref{maxQ} in \eqref{numupper} or in \eqref{Pnumupper} and get that $\langle \mathbb{B}_1 \rangle = 0$ or $G(y^{0}=1,P_{obs})=\frac{1}{2}$, which coincides with the certification of one perfect local random bit for input $y_0 = 1$ on Bob's side for the maximal violation of $I_{\theta}$. 
Since the probability distribution of maximal violation is unique, the point is necessarily an extreme point \cite{Acin2012}, so we can directly use the observed probability distribution $P_{obs}$ to bound the eavesdropper's predictive power (as an extreme point allows only for one decomposition: itself). \\
\\

Let us finally consider the case of $\alpha=1$ and $\beta = 2\cos(2\theta)/[1+\sin^2(2\theta)]^{1/2}$, which results in the Bell inequality \eqref{Eq:Ibeta} considered in the main text.  Figure \ref{Fig:weak} presents the bound (\ref{Pnumupper}) for three values of $\theta$, in particular for $\theta=\pi/4$ which corresponds to the CHSH Bell inequality. This should provide one with an intuition of how close quantitatively to the maximal violation $I_{\theta}^{max}$ the observed value $I_{\theta}$ should be in order to get close to one perfect local bit of randomness ($G(y=1,P_{obs}) \rightarrow 1/2$) for a state with a given angle $\theta$.

\section{The amount of certified randomness as a function of the strength of the measurement}\label{Weak}

We know already that the violation of a Bell inequality certifies the existence of randomness in the outcomes of the measurements. The other way is also true, namely that if the solution of the optimization problem \eqref{sumupto} gives a solution $G(y^{0},P_{obs}) < 1$ then the observed behavior $P_{obs}$ is necessarily nonlocal. On a purely qualitative level, certified randomness in the outcomes is equivalent to nonlocal correlations.\\

In this section we analyze with the help of numerical tools the dependency of the certified randomness from the violation of the family of Bell inequalities \eqref{Eq:Ibeta} on the strength parameter $\xi$ of the measurements $\hat{\sigma}_x(\xi) = cos(2\xi)\sigma_x$ \eqref{MeasKraus}. For example, what is the maximal value of the parameter $\xi$ -- i.e. the minimal strength of the measurement --such that we can generate nonlocal correlations (and thus randomness) from this measurement on an entangled state of the form $\ket{\psi(\theta)}$ \eqref{Eq:qubits}? Do less entangled states need stronger measurement to unveil their nonlocal behavior?\\

To answer these questions, we have been using semi-definite programing (SDP) techniques as explained in \cite{bancal2014more,Nieto} to obtain numerical upper bounds on the guessing probabilities \eqref{sumupto}. One can find the computational details -- presented in a pedagogical way -- online at \url{https://github.com/peterwittek/ipython-notebooks/blob/master/Unbounded_randomness.ipynb}. Here we work in the standard scenario with only one measurement $n = 1$ in the sequence. We used states of the form \eqref{Eq:qubits}:
\begin{equation}\label{psitheta2}
\ket{\psi(\theta)} = \cos(\theta)\ket{00}+\sin(\theta)\ket{11}
\end{equation}
and measurements \eqref{Eq:Ibeta2}:
\begin{equation}\label{measur}
\begin{split}
\mathbb{A}_0 = \cos\mu\,\sigma_z + \sin\mu\,\sigma_x, \hspace{1cm} \mathbb{B}_0 = \sigma_z, \\
\mathbb{A}_1 = \cos\mu\,\sigma_z - \sin\mu\,\sigma_x, \hspace{1cm} \mathbb{B}_1 = \hat{\sigma}_x(\xi) = cos(2\xi)\sigma_x,
\end{split}
\end{equation}
where $tan(\mu) = sin(2\theta)$. These measurements correspond to the ones in our scheme for an unbounded amount of randomness and where the second measurement $y = 1$ of $B$ is the tunable version $\hat{\sigma}_x(\xi)\equiv\{M_{+1}^{\dagger}M_{+1},M_{-1}^{\dagger}M_{-1}\}$ of Eq. \eqref{MeasKraus}: 
\begin{equation}
M_{\pm1}(\xi)=\cos\xi|\pm\rangle\!\langle\pm|+\sin\xi|\mp\rangle\!\langle\mp|,
\end{equation}
with $\xi \in [0,\frac{\pi}{4}]$. For example, if the parameter $\xi = 0$, the four (projective) measurements in Eq. \eqref{measur} on any quantum state $\ket{\psi(\theta)}$ with angle $\theta$ \eqref{psitheta2} generates a behavior $P_{obs}^{\theta}$ leading to the maximal violation of the inequality $I_{\theta}$ \eqref{Eq:Ibeta} for the same value of $\theta$. This implies that extremal nonlocal correlations are generated and from the results of \cite{Acin2012} we know that one perfect random bit -- equivalently $G(y^{0}=1,P^{\theta}_{obs}) = \frac{1}{2}$ -- is produced. This corresponds to the strongest (projective) version of the measurements. Now, as we increase the parameter $\xi > 0$ of $B$'s $y = 1$ measurement, $\hat{\sigma}_x(\xi)$ gets weaker, the generated correlations cease to be extremal and less than one random bit is produced. At some point, at a particular value $\xi^{\theta}_{\textrm{max}}$ the measurement of $B$ is so weak that we expect the generated correlations to become local. This exact value might depend on the amount of entanglement $\theta$ in the state. The bounds obtained by SDP indicate that this dependency on the angle $\theta$ of the maximal value $\xi^{\theta}_{\textrm{max}}$ is relatively small. As we vary the angle $\theta$, the minimal required strength of the measurement to generate a nonlocal behavior $P^{\theta}_{obs}$ stays within a narrow interval: $\xi^{\theta}_{\textrm{max}} \in [0.519,0.576]$ for $\theta \in [\frac{\pi}{32},\frac{\pi}{4}]$.\\

We now present the results in the form of a graph (see Fig.\ref{Fig:weak}). A complete tables with our results for the different states and bounds on the guessing probabilities can be found in the appendices \ref{APP5:SDP_Programs}.

\begin{figure}[h!]
\scalebox{0.6}{\includegraphics{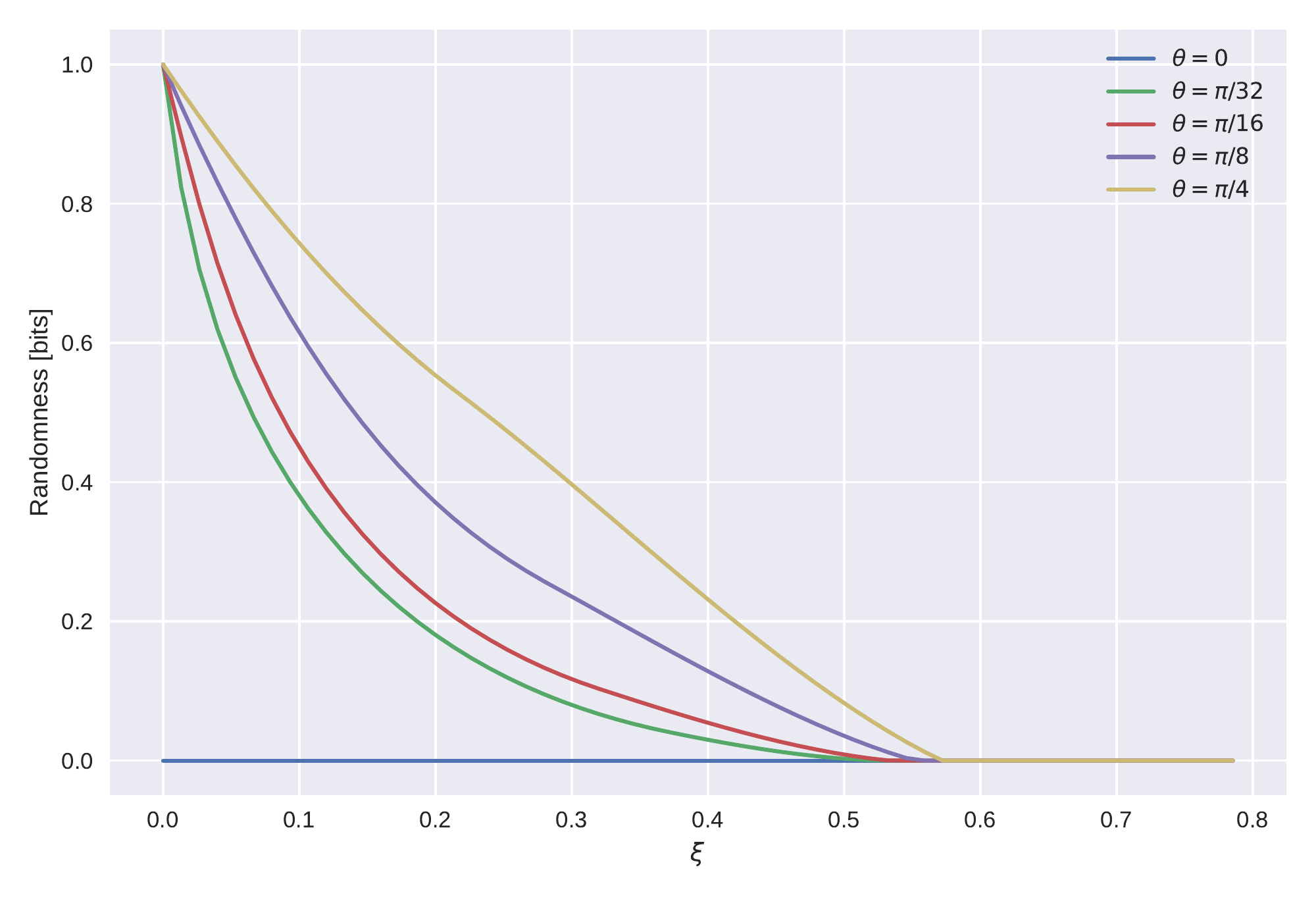}}
\caption{\label{Fig:weak}
Lower bounds on the amount of randomness certified from the quantum state \eqref{Eq:qubits} with angles $\theta = 0,\frac{\pi}{32},\frac{\pi}{16},\frac{\pi}{8},\frac{\pi}{4}$ as function of the strength of the measurement $\xi$. The measurement is projective for $\xi = 0$ -- which certifies the maximal amount of randomness -- and is non interacting with the system when $\xi = \frac{\pi}{4}$. It is intriguing to see that for the cases of $\frac{\pi}{32} \leq \theta \leq \frac{\pi}{4}$ considered the generated behavior become local in a small interval $\xi_{\textrm{max}} \in [0.519,0.576]$.}
\end{figure}

As expected the amount of certified randomness for each state $\ket{\psi(\theta)}$ is one bit when the measurement is projective (for $\xi = 0$) as the correlations are the extremal ones described in \cite{Acin2012} regardless of the entanglement $\theta$ in the state.
As $\xi$ increases the lower bounds on the certified randomness rapidly decreases, with  a more rapid decrease for smaller $\theta$. Interestingly, and up to (high) numerical precision, for all values of $\theta$ the bounds reach zero certified randomness around the same value $\xi_{\textrm{max}} \in [0.519,0.576]$. This indicates, again up to numerical precision, that all the generated $P_{obs}^{\theta}$ become local -- or stop generating randomness -- around this critical value.\\

In the end, we are interested primarily in the amount of certified randomness from $P_{obs}^{\theta}$ \textit{close to} the maximal violation of $I_{\theta}$, corresponding to $\xi \rightarrow 0$. There, the SDP solutions indicate that the correlations resisting the best to the weakening of the measurement $\xi > 0$ are the ones coming from the measurements made on the maximally entangled state. Indeed, if the bounds are close to the actual values of certified randomness it is quite clear from the numerical results that the more the state is entangled ($\theta \rightarrow \frac{\pi}{4}$) the better it resists. The less entangled states ($\theta \rightarrow 0$) appear to generate exponentially less randomness when the parameter $\xi$ increases, or equivalently when the correlations cease to be extremal. This tells us that even though our scheme certifies an unbounded amount of randomness from states $\ket{\psi(\theta)}$ with any nonzero amount of entanglement, i.e. any $\theta > 0$, it is preferential from a practical point of view to use the maximally entangled state as the initial state.


\section{Conclusion}\label{Conclusions} We have presented a scheme for certifying an unbounded amount of random bits from pairs of entangled qubits in the scenario where one of the qubits is subjected to a sequence of measurements. The measurements do not completely destroy the entanglement but map the state to another pure entangled two-qubit state (with reduced entanglement). Our main result made use of the fact that every measurement in Bob's sequence generated an almost-maximally non-local output distribution (in the sense of violating some Bell inequality almost maximally). In Ref. \cite{Silva}, a sequence of non-local correlations is obtained from pairs of qubits, showing that the nonlocality of a state can be shared between many parties. While it also considers sequences of measurements, one can show that the correlations obtained in their work do not generate more certified randomness than the simple standard single measurement scenario. Indeed, the maximum of randomness is achieved when all but one measurements do not interact with the particle and their scheme is thus optimal when coinciding with a single measurement one. In our work, we overcome this limitation by producing (almost) extremal correlations for each measurement in the sequence, which is a fundamental property of potential further use for many other device-independent quantum information tasks (in particular for randomness certification).
Our work is in many respects a proof-of-principle result: First, it requires an exponentially increasing number of measurements on Alice's side, namely $\sum_{i=1}^{n}2^i=2(2^n-1)$ measurement choices for $n$ measurements in the sequence. Second, the result is based on a continuity argument and there is no control on the noise robustness. All these issues deserve further investigation. Finally, it is worth exploring how to design device-independent randomness generation protocols involving sequences of measurements. However, the sequential scenario is much more demanding from an implementation point of view, because it requires quantum non-demolition measurements. It is then unclear whether with present or near future technology sequential protocols will provide a significant practical advantage over simpler protocols based on standard Bell tests. However, the first experimental works observing non-local correlations in the sequential scenario have recently been reported~\cite{schiavon2017three,hu2016experimental}. In any case, the main implications of our work are fundamental: It shows that pairs of pure entangled qubits are potentially unbounded sources of certifiable random bits when performing sequences of measurements on it.\\

We have also provided numerical results that gives us an insight on the resistance to imperfections of a potential protocol that implements our scheme. For a single measurement in the sequence, we have given numerical bounds on how the certified randomness diminishes as the generated correlations cease to be extremal. Second, we have also explored how the certified randomness diminishes when the strength of the measurement is lowering. This allows us to expect that any potential protocol trying to implement our scheme for a finite amount of randomness starting from an entangled system has an advantage in using the maximally entangled one. It is clear from our numerical results that this state offers the best resistance to imperfections. So, while it is true that even arbitrarily little entangled states are a source of unbounded certified randomness, more entanglement offers an advantage in terms of resistance to imperfections.\\

It would also be interesting to explore whether an unbounded amount of randomness can be obtained versus a post-quantum adversary $E$, only constrained by the no-signaling condition, trying to guess the outcomes of the measurements. Or, on the contrary, is the amount of certified randomness against no-signaling adversaries bounded also in the sequential scenario? Our conjecture is that the amount of randomness that can be certified is limited in this case. Indeed, the fact that the no-signaling set -- consisting of all correlations constrained only by the no-signaling conditions -- does not have a continuous set of extremal points (it is a polytope) makes it impossible to obtain a sequence of extremal probability distributions in a sequence as the one that we could obtain in the quantum case. A different approach thus needs to be taken. It is really the fact that the quantum set has curved boundaries made of extremal quantum behaviors that allowed to derive the results of this paper.

\section{Acknowledgments}
This work is supported by the ERC CoG QITBOX and AdG OSYRIS, the AXA Chair in Quantum Information Science, Spanish MINECO (QIBEQI and SEV-2015-0522), Fundaci\'on Cellex, Generalitat de Catalunya (SGR 875 and Cerca Program). M.J. acknowledges support from the Marie Curie COFUND action through the ICFOnest program. R. A. acknowledges funding from the European Union's Horizon 2020 research and innovation programme under the Marie Sk\l{}odowska-Curie grant agreement No 705109. M.J.H. acknowledges support from the EPSRC (through the NQIT Quantum Hub), the FQXi Large Grant Thermodynamic vs information theoretic entropies in probabilistic theories, and the hospitality of the Department of Computer Science at the University of Oxford. P.W. acknowledges computational resources granted by the High Performance Computing Center North (SNIC 2015/1-162 and SNIC 2016/1-320).



\appendix
\section{The guessing probability}\label{APP1:PGgen}

We start our appendices with the following discussion, which is a summary of the work done in deriving the device-independent guessing probability (DIGP) \cite{Randomness, Acin2012,Nieto,DeLaTorre}. A conditional probability distribution that is the outcome distribution for some measurement on a quantum state is called a quantum distribution. For example, a distribution $P$ with elements $p(ab|xy)$ is quantum if there exist at least one quantum state, i.e., a positive semi-definite hermitian unit trace matrix $\rho$ and at least one set of measurements, i.e., a set of positive semi-definite hermitian matrices $M_{a|x}$, $M_{b|y}$ satisfying $\sum_a M_{a|x}=\sum_b M_{b|y}=1$ such that $p(ab|xy)=Tr(M_{a|x}\otimes{M_{b|y}}\cdot \rho)$. We will often abuse notation and refer to a distribution by its elements $p(ab|xy)$ when there is no confusion in doing so.

The set $\mathcal{\mathcal{\mathcal{Q}}}$ of quantum distributions is convex and a distribution in $\mathcal{\mathcal{\mathcal{Q}}}$ that cannot be decomposed as a convex combination of other distributions is called {\it extremal} in $\mathcal{\mathcal{\mathcal{Q}}}$. For a non-extremal distribution $P(ab|xy)$ there is in general more than one possible convex decomposition.

A non-extremal distribution $p(ab|xy)$ with a convex decomposition $p(ab|xy)=\sum_{\lambda}q_{\lambda}p_{\lambda}(ab|xy)$ can be constructed by sampling the different
distributions $p_{\lambda}(ab|xy)$ with probability $q_{\lambda}$. In this case knowledge about the convex decomposition chosen changes the ability of an eavesdropper to correctly guess the outcomes $a$ and/or $b$.

Without knowledge of the decomposition, or for extremal distributions, the probability of correctly guessing the outcome of measurement $y^0$ is $\max_{b}p(b|y^0)$, the probability of the most likely outcome.
With knowledge of the decomposition $p(ab|xy)=\sum_{\lambda}q_{\lambda}p_{\lambda}(ab|xy)$, the probability is larger or equal to $\max_{b}p(b|y^0)$
\begin{eqnarray}\label{noos}
\sum_{\lambda}q_{\lambda}\max_{b}p_{\lambda}(b|y^0)\geq \max_{b}\sum_{\lambda}q_{\lambda}p_{\lambda}(b|y^0)=\max_{b}p(b|y^0).\end{eqnarray}
For a given observed non-extremal distribution $P_{\textrm{obs}}$, it is possible that it was produced by an agent Eve that has larger predictive power than an agent which only observes the outcomes.

We now want to consider the optimal probability for the agent Eve to correctly guess an outcome $b$ of measurement $y^0$ given a distribution $p_{obs}(ab|xy)$ and control over its decomposition in extremal points. If the set of quantum distributions is closed there exist one or several optimal ways to decompose the given distribution that maximizes this probability. If the set is not closed but open or semi-open, there may not exist a maximum and the relevant quantity is instead the supremum value of Eves probability to correctly guess the outcome. Since $\max_{b}p(b|y^0)$ is a continuous function on the set of probability distributions it follows that the supremum value of $\sum_{\lambda}q_{\lambda}\max_{b}p_{\lambda}(b|y^0)$ as a function of all possible decompositions, indexed by $\lambda$, on an open or semi-open set of distributions is the same as the maximum value on the closure of the set. Therefore, in this case we can consider the closure of the set and express the 
 probability as an optimization over the extremal points of this closed set.

With this disclaimer, the maximal probability for the agent Eve to correctly guess an outcome $b$ of measurement $y^0$ given a distribution $p_{obs}(ab|xy)$ and control over the decomposition is the DIGP $G(y^0,P_{\textrm{obs}})$
\begin{equation}\label{gutt}
G(y^0,P_{\textrm{obs}})=\underset{q_{\lambda},p_{\lambda}(ab|xy)}{\textrm{max}}\sum_{\lambda}q_{\lambda}\max_{b}p_{\lambda}(b|y^0).
\end{equation}
where $\lambda$ is labelling the convex decompositions of $p_{\textrm{obs}}(ab|xy)$ in terms of extremal distributions $p_{\lambda}(ab|xy)$.
Note that if $\mathcal{\mathcal{\mathcal{Q}}}$ is not closed a given extremal point may not belong to the set but only to its closure.
For any open interval of $\mathcal{\mathcal{\mathcal{Q}}}$ the function $G(y^0,P_{\textrm{obs}})$ is a concave function \cite{Randomness}.
Therefore this kind of maximization is called a {\it concave roof} construction.

The guessing probability can be approximated by a hierarchy of semidefinite programming (SDP) relaxations~\cite{Nieto,bancal2014more}. We used Ncpol2sdpa~\cite{wittek2015ncpol2sdpa} to generate the relaxations for verifying some of the analytical results. We relied on the arbitrary-precision variant of the SDPA family of solvers~\cite{yamashita2003implementation} for obtaining important numerical values, and the solver Mosek\footnote{\url{http://mosek.com/}} in all other cases.

\section{Continuity of the guessing probability in interior and extremal points of $\mathcal{Q}$}
\label{APP2:ContinuityProof}

The guessing probability as a function on the space of probability distributions is not everywhere continuous. An example of this is that the family of Bell-inequalities of Ref. \cite{Acin2012} that certifies one bit of randomness for measurements on a state with arbitrarily little entanglement. The probability distribution corresponding to such a state and the measurements in Eq. \ref{Eq:Ibeta2} has $G(y^0,P_{obs})=1/2$ and is at the same time arbitrarily close to a distribution corresponding to measurements on a product state with $G(y^0,P_{obs})=1$, i.e., a distribution which can be prepared by a local deterministic procedure. There is thus a discontinuity where the guessing probability jumps from $1/2$ to $1$.  The key to understanding this discontinuity is that the local deterministic distribution is not extremal while the quantum distribution in the neighbouring point is extremal. As seen in Eq. \ref{noos}, the guessing probability is given by different functions depend
 ing on whether a distribution can be decomposed into other distributions or not, i.e., if it is extremal or not. This means discontinuities can appear at the boundary between extremal points and non-extremal points.

We will now show that discontinuities can {\it only} appear at such boundaries between extremal and non-extremal points in the boundary $\partial{\mathcal{Q}}$ of the quantum set $\mathcal{Q}$. To do this we use the property of the guessing probability described in Eq. \ref{noos}, together with some general properties of concave functions and in particular concave roof constructions.

We want to show that the following propositions are true:

\begin{proposition}\label{prop1}
The function $G(y^0,P_{\textrm{obs}})$ on the set of quantum distributions $\mathcal{Q}$ is continuous in the interior of $\mathcal{Q}$.
\end{proposition}
\begin{proposition}\label{prop2}
The function $G(y^0,P_{\textrm{obs}})$ is continuous in any extremal point of $\mathcal{Q}$.
\end{proposition}

Proposition \ref{prop1} is trivial.
The guessing probability $G(y^0,P_{\textrm{obs}})$ is concave by definition and any concave function is continuous on an open subset of its domain \cite{rock}. In particular this means that $G(y^0,P_{\textrm{obs}})$ is continuous in the interior of $\mathcal{Q}$. Note that if $\mathcal{Q}$ is open, i.e. has no boundary, there can thus not exist any discontinuity.

To address proposition \ref{prop2} we consider the restriction $G(y^0,P_{\textrm{obs}})^{\partial{\mathcal{Q}}}$ of $G(y^0,P_{\textrm{obs}})$ to the boundary $\partial{\mathcal{Q}}$ of the quantum set. First we note that the function $G(y^0,P_{\textrm{obs}})^{\partial{\mathcal{{Q}}}}$ by definition is continuous on any open set of extremal points since $\max_{b}p(b|y)$ is a continuous function.
Next we observe that the boundary $\partial{\mathcal{{Q}}}$ can be decomposed into a collection of open sets of extremal points and a collection $\{S_i\}$ of closed connected possibly overlapping sets where each set is the closure of a maximal open connected subset. A maximal open connected subset $M$ of the non-extremal points is an open set such that any other open connected set of non-extremal points which contains $M$ is $M$ itself. Therefore, each set $S_i$ is the convex hull of the set of extremal points in its closure.

Any closed set $S_i$ has a boundary $\partial S_i$ with the rest of $\partial{\mathcal{{Q}}}$ which can be decomposed in the same way into open sets of extremal points and closed connected sets $S_{ij}$ that are closures of maximal open connected sets of non-extremal points.
The boundary  $\partial S_{ij}$ of $S_{ij}$ with the rest of $\partial S_i$ is in turn decomposable in the same way.

Continuing this successive decomposition of the boundary $\partial \mathcal{Q}$ we will eventually reach sets $S_{ijk\dots}$ that are one dimensional simplexes, or alternatively sets with only extremal points in the boundary. On sets of these two types $G(y^0,P_{\textrm{obs}})$ is a continuous function. To see this we introduce the following terminology, and use a theorem from Ref. \cite{buci}.

A function for which all discontinuities are such that the function takes the higher value at a closed set and the lower value at an open set is called {\it upper semi-continuous}.

The function $G(y^0,P_{\textrm{obs}})^{S}$ defined on a closed convex set $S$ can be viewed as an extension of $G(y^0,P_{\textrm{obs}})^{\partial{S}}$ to the interior of $S$. This extension is called the {\it concave roof extension}.

\begin{theorem}\label{theo1}{Let C be a compact set and $K=co(C)$ be the convex hull of C. If $F:C\to{\mathbb{R}}$ is bounded, upper semi-continuous, and concave on C, then the concave roof extension $\hat{F}:K\to\mathbb{R}$ of F to K is upper semi-continuous \cite{buci}.}
\end{theorem}

The guessing probability is bounded and concave by definition.
If the boundary of $S$ has only extremal points it follows that $G(y^0,P_{\textrm{obs}})^{\partial{S}}$ is continuous in $\partial{S}$ and by theorem \ref{theo1} $G(y^0,P_{\textrm{obs}})^{S}$ is upper semi-continuous on $S$. Moreover, since $G(y^0,P_{\textrm{obs}})^{S}$ is concave it cannot have an upper semi-continuous discontinuity between the boundary and the interior.
If $S$ is a one-dimensional simplex we can, if necessary, restrict the domain of the guessing probability to a one dimensional subspace and make the same argument.

Next we consider discontinuities between $S$ and an open set of extremal points.

\begin{lemma}\label{theo2}Any discontinuity of $G(y^0,P_{\textrm{obs}})$ between a closed set and an open set of extremal points is upper semi-continuous.
\end{lemma}
\begin{proof}
If the boundary point of the closed set is extremal the $G(y^0,P_{\textrm{obs}})$ is continuous since $\max_{b}p(b|y^0)$ is continuous. Next consider a non-extremal boundary point of the closed set. $G(y^0,P_{\textrm{obs}})$ in the non-extremal point is always greater or equal to $\max_{b}P(b|y^0)$ by Eq. \ref{noos}. Thus any discontinuity is upper semi-continuous.
\end{proof}
If there is a discontinuity of $G(y^0,P_{\textrm{obs}})$ on the boundary of $S$ it is, by lemma \ref{theo2} , upper semi-continuous and at a set of non-extremal points.

By repeated application of Theorem \ref{theo1} and lemma \ref{theo2} we can conclude that $G(y^0,P_{\textrm{obs}})^{\partial{\mathcal{{Q}}}}$ is upper semi-continuous on $\partial{\mathcal{{Q}}}$ and that $G(y^0,P_{\textrm{obs}})$ is upper semi-continuous on ${\mathcal{{Q}}}$. Since $G(y^0,P_{\textrm{obs}})$ is concave there cannot be an upper semi-continuous discontinuity between the boundary $\partial{\mathcal{{Q}}}$ and the interior of $\mathcal{{{Q}}}$. Thus the only discontinuities are between non-extremal points in closed subsets of $\partial{\mathcal{{Q}}}$ and extremal points in open subsets of $\partial{\mathcal{{Q}}}$.

\section{Bounds on the guessing probability as a function of a Bell inequality: Continuity at a unique point of maximal violation}
\label{APP6:ContAtPextUnique}

We have described the guessing probability as a function on set of quantum distributions, but it is sometimes useful to consider it as a function of the violation of some given Bell inequality $I$.
A Bell expression is a linear function on the space of distributions and the set of distributions for which it takes a given value $t$ is a hyper-plane $H_{t}$. The different values of the Bell expression thus defines a family of parallel hyperplanes.

On each hyperplane $H_{t}$ we can consider the restriction $G(y^0,P_{\textrm{obs}})_t$ of $G(y^0,P_{\textrm{obs}})$ to the intersection of $H_{t}$ with $\mathcal{{{Q}}}$ and take its maximum $\max G(y^0,P_{\textrm{obs}})_t$ on this intersection. This maximum is the highest probability for Eve to guess the outcome of $y^0$ for any distribution $P\in{\mathcal{{Q}}}$ such that $I(P)=t$. The function $\max G(y^0,P_{\textrm{obs}})_t$ can have a discontinuity at $t=t_c$ only if $H_{t_c}$ intersects with a point in $\mathcal{{{Q}}}$ at which ${G}(y^0,P_{\textrm{obs}})$ is discontinuous.

Let us consider a Bell expression $I$ and its maximal value $t_{max}$ on $\mathcal{\mathcal{\mathcal{Q}}}$.
If the intersection of $H_{t_{max}}$ and $\mathcal{\mathcal{\mathcal{Q}}}$ is a single extremal point
it follows from Propositions \ref{prop1} and \ref{prop2} that there is a $t_c\neq t_{max}$ such that for the range  $t_{c}\leq t\leq t_{max}$ for which $\max G(y^0,P_{\textrm{obs}})_t$ is a continuous function of $t$.

If the intersection of $H_{t_{max}}$ and $\mathcal{\mathcal{\mathcal{Q}}}$ contains more than one extremal point it also contains a set of non-extremal points of $\partial \mathcal{\mathcal{\mathcal{Q}}}$ and $G(y^0,P_{\textrm{obs}})$ could have a discontinuity between this set and an open set of extremal points. This discontinuity could lead to a discontinuity of the function $\max G(y^0,P_{\textrm{obs}})_t$ at $t_{max}$.

\section{Guessing probability for a sequence}\label{APP3:PGforSequ}

So far, we have discussed the continuity properties of the guessing probability in the standard scenario, where one single measurement $M_{a|x}$ is made on Alice's side and $M_{b|y}$ on Bob's. The goal of this section is to extend these properties to the case where sequential measurements $M_{a_i|x_i}$ and $M_{b_i|y_i}$ are performed by each party, where $i$ labels the position of a particular measurement in the sequence.

Let us consider a sequence of measurements $\hat{\sigma}(\xi_i)$ chosen by Bob and denote $(\xi_1,\xi_{2},\dots{,\xi_{n}})\equiv{\vec{\xi}}$. The convex decomposition of the observed outcome distribution that gives Eve optimal probability to correctly guess the sequence of outcomes $\vec{b}_n$ of the measurements $(y^0_1,y^0_{2},\dots{,y^0_{n}})\equiv{\vec{y}^0_n}$ is a function of $\vec{\xi}$. The guessing probability $G(\vec{y}^0_n,P_{\textrm{obs}})$ is thus given by
\begin{equation}\label{gutt2}G(\vec{y}^0_n,P_{\textrm{obs}})=\sum_{\lambda_{\bar{\xi}}}q_{\lambda_{\vec{\xi}}}\max_{\vec{b}_{n}}p_{\lambda_{\vec{\xi}}}(b_1|y^0_1)\cdot p_{\lambda_{\vec{\xi}}}(b_2|y^0_2,y^0_1,b_1)\dots p_{\lambda_{\vec{\xi}}}(b_n|\vec{y}^0_n\vec{b}_{n-1}).
\end{equation}
where the extremal distributions $p_{\lambda_{\vec{\xi}}}(b_n|y_n\dots)$ and weights $q_{\lambda_{\vec{\xi}}}$ of the optimal convex decomposition are functions of $\vec{\xi}$ as indicated by the index $\lambda_{\vec{\xi}}$.
Let us assume that a term which appears in the convex combination is
\begin{equation}\label{kutt}
q_{\lambda_{\vec{\xi}}}p_{\lambda_{\vec{\xi}}}(b_1|y^0_1)\dots p_{\lambda_{\vec{\xi}}}(b_n|{\vec{y}^0_n}\vec{b}_{n-1}).
\end{equation}
Thus we assume that it corresponds to the most probable sequence of outcomes $\vec{b}_{n}$ for a specific distribution indexed by $\lambda_{\vec{\xi}}$.

Given that Eve has chosen the optimal convex decomposition for guessing the outcomes of $\vec{y}^0_n$ we consider her probability of correctly guessing the outcome of $y^0_m$ for $1\leq m\leq{n}$ given a particular sequence of previous outcomes $\vec{b}_{m-1}$. It is given by
\begin{equation}\label{klut}
\sum_{\lambda_{\vec{\xi}}}k_{\lambda_{\vec{\xi}}}\max_{b_m}p_{\lambda_{\vec{\xi}}}(b_m|{\vec{y}^0_m}\vec{b}_{m-1}),
\end{equation}
where $k_{\lambda_{\vec{\xi}}}$ is the probability that the distribution indexed by $\lambda_{\vec{\xi}}$ will be sampled given the sequence of previous outcomes $\vec{b}_{m-1}$
\begin{equation}\label{nut}
k_{\lambda_{\vec{\xi}}}=\frac{q_{\lambda_{\vec{\xi}}}p_{\lambda_{\vec{\xi}}}(b_1|y^0_1)\dots p_{\lambda_{\vec{\xi}}}(b_{m-1}|{\vec{y}^0_{m-1}}\vec{b}_{m-2})}{\sum_{\lambda_{\vec{\xi}}}q_{\lambda_{\vec{\xi}}}p_{\lambda_{\vec{\xi}}}(b_1|y^0_1)\!\dots\! p_{\lambda_{\vec{\xi}}}(b_{m-1}|{\vec{y}^0_{m-1}}\vec{b}_{m-2})}.
\end{equation}

The probability in Eq. \ref{klut} is larger or equal to $1/d_m$, where $d_m$ is the number of possible outputs $b_m$, but is lower or equal to $G(y^0_m,P_{\textrm{obs}})$, the
maximal probability that Eve could guess the outcome of $y^0_m$ correctly given that she had chosen an optimal strategy for this and not the optimal strategy for guessing the outcomes of the sequence $\vec{y}^0_n$. Thus if $G(y^0_m,P_{\textrm{obs}})$ is close to $1/d_m$ so is the expression in Eq. \ref{klut}.

\section{Arbitrarily close to $n$ random bits for $n$ measurements}\label{APP4:ArbitrarySequ}

We want to prove that $G({\vec{y}^0_n},P_{\textrm{obs}})$ can be made arbitrarily close to $2^{-n}$ by making $G(y^0_m,P_{\textrm{obs}})$ sufficiently close to 1/2 for each $1\leq m\leq n$.

The proof relies on the fact that if a convex combination of a collection of numbers $x_i$ equals $a$, i.e., $\sum_ik_ix_i=a$ where $\sum k_i=1$, and if $x_i\geq{a}$ for each $i$, it follows that for every $i$ either $k_i=0$ or $x_i=a$.

From this follows that when $G(y^0_m,P_{\textrm{obs}})$ is very close to 1/2 either $\max_{b_m}p_{\lambda_{\vec{\xi}}}(b_m|{\vec{y}^0_m}\vec{b}_{m-1})$ in Eq. \ref{klut} is very close to 1/2 or $k_{\lambda_{\vec{\xi}}}$ is very close to zero for each $\lambda_{\vec{\xi}}$.
To see this more clearly we construct the following bound

\begin{eqnarray*}k_{\lambda_{\vec{\xi}}}\max_{b_m}p_{\lambda_{\vec{\xi}}}(b_m|{\vec{y}^0_m}\vec{b}_{m-1})&\leq&G(y^0_m,P_{\textrm{obs}})-\sum_{\lambda'\neq\lambda}k_{\lambda'_{\vec{\xi}}}\max_{b_m}p_{\lambda'_{\vec{\xi}}}(b_m|{\vec{y}^0_m}\vec{b}_{m-1})\nonumber\\
&\leq&{G(y^0_m,P_{\textrm{obs}})-1/2(1-k_{\lambda_{\vec{\xi}}})}\end{eqnarray*}
where we used $\max_{b_m}p_{\lambda'_{\vec{\xi}}}(b_m|{\vec{y}^0_m}\vec{b}_{m-1})\geq{1/2}$ for each $\lambda'_{\vec{\xi}}$ and $\sum_{\lambda'\neq\lambda}k_{\lambda'_{\vec{\xi}}}=1-k_{\lambda_{\vec{\xi}}}$.
It follows that
\begin{eqnarray*}
G(y^0_m,P_{\textrm{obs}})-1/2\geq k_{\lambda_{\vec{\xi}}}[\max_{b_m}p_{\lambda_{\vec{\xi}}}(b_m|{\vec{y}^0_m}\vec{b}_{m-1})-1/2],
\end{eqnarray*}
and given Eq. \eqref{nut} this implies

\begin{eqnarray*}
G(y^0_m,P_{\textrm{obs}})-1/2\geq q_{\lambda_{\vec{\xi}}}p_{\lambda_{\vec{\xi}}}(b_1|y^0_1)\dots p_{\lambda_{\vec{\xi}}}(b_{m-1}|{\vec{y}^0_{m-1}}\vec{b}_{m-2}){[\max_{b_m}p_{\lambda_{\vec{\xi}}}(b_m|{\vec{y}^0_n}\vec{b}_{m-1})-1/2]}.
\end{eqnarray*}
Thus for sufficiently small $G(y^0_m,P_{\textrm{obs}})-1/2$ either $\max_{b_m}p_{\lambda_{\vec{\xi}}}(b_m|{\vec{y}^0_m}\vec{b}_{m-1})-1/2$ can be made arbitrarily small,
or the probability $q_{\lambda_{\vec{\xi}}}p_{\lambda_{\vec{\xi}}}(b_1|y^0_1)\dots p_{\lambda_{\vec{\xi}}}(b_{m-1}|{\vec{y}^0_{m-1}}\vec{b}_{m-2})$ that the distribution labelled by $\lambda_{\vec{\xi}}$ is sampled when $y^0_m$ is measured is made arbitrarily small.

The argument can be made for any $B_m$. For $B_1$, it follows that either $p_{\lambda_{\vec{\xi}}}(b_1|y^0_1)$ is made arbitrarily close to $1/2$ or $q_{\lambda_{\vec{\xi}}}$ is made arbitrarily close to $0$.
For $B_2$, it follows that either $p_{\lambda_{\vec{\xi}}}(b_2|y^0_2y^0_1b_1)$ is made arbitrarily close to $1/2$ or $q_{\lambda_{\vec{\xi}}}p_{\lambda_{\vec{\xi}}}(b_1|y^0_1)$ is made arbitrarily close to zero. Given the second option and that $p_{\lambda_{\vec{\xi}}}(b_1|y^0_1)$ is made arbitrarily close to $1/2$ it is implied that that $q_{\lambda(\vec{\xi})}$ is made arbitrarily close to 0. If on the other hand $p_{\lambda_{\vec{\xi}}}(b_1|y^0_1)$ is not very close to $1/2$ it follows that $q_{\lambda_{\vec{\xi}}}$ is made arbitrarily close to zero by the preceding argument.

By induction it is clear that either the term in Eq. \ref{kutt} satisfies that $p_{\lambda_{\vec{\xi}}}(b_1|y^0_1)\dots p_{\lambda_{\vec{\xi}}}(b_n|{\vec{y}^0_n}\vec{b}_{n-1})$ can be made arbitrarily close to ${2^{-n}}$ or alternatively $q_{\lambda_{\vec{\xi}}}$ is made arbitrarily small. Since the same is true for every $\lambda_{\vec{\xi}}$ in Eq. \ref{gutt2} it follows that $G({\vec{y}^0_n},P_{\textrm{obs}})$ can be made arbitrarily close to ${2^{-n}}$.

Note that the above argument can be straightforwardly extended to the case where the number of outputs $d_i$ for each $B_i$ can be different from 2. Thus, in this case $G({\vec{y}^0_n},P_{\textrm{obs}})$ can be made arbitrarily close to $\prod_{i=1}^nd_i^{-1}$ by making $G(y^0_m,P_{\textrm{obs}})$ sufficiently close to $1/d_m$ for each $1\leq m\leq n$.

\section{Our programs to obtain lower bounds on the certified randomness}\label{APP5:SDP_Programs}

In this section of the appendices we give the tables of results for section \ref{Weak}. We remind the reader that the computational details -- exposed in a pedagogical way -- of our results can be found online at:\\

\url{https://github.com/peterwittek/ipython-notebooks/blob/master/Unbounded_randomness.ipynb}\\


\begin{table}
\parbox{.45\linewidth}{
\centering
\caption{$\theta = \frac{\pi}{4}$, the maximally entangled state}
\begin{tabular}{|| c c ||}
\hline
$\xi$ & $\sharp$ random bits\\
\hline
0.000 & 1.000\\
0.013 & 0.962\\
0.027 & 0.925\\
0.040 & 0.890\\
0.053 & 0.855\\
0.067 & 0.822\\
0.080 & 0.790\\
0.093 & 0.759\\
0.106 & 0.729\\
0.120 & 0.700\\
0.133 & 0.673\\
0.146 & 0.647\\
0.160 & 0.622\\
0.173 & 0.598\\
0.186 & 0.575\\
0.200 & 0.554\\
0.213 & 0.533\\
0.226 & 0.514\\
0.240 & 0.494\\
0.253 & 0.473\\
0.266 & 0.452\\
0.280 & 0.430\\
0.293 & 0.409\\
0.306 & 0.387\\
0.319 & 0.365\\
0.333 & 0.342\\
0.346 & 0.320\\
0.359 & 0.298\\
0.373 & 0.276\\
0.386 & 0.254\\
0.399 & 0.233\\
0.413 & 0.211\\
0.426 & 0.190\\
0.439 & 0.170\\
0.453 & 0.150\\
0.466 & 0.130\\
0.479 & 0.111\\
0.493 & 0.093\\
0.506 & 0.075\\
0.519 & 0.058\\
0.532 & 0.042\\
0.546 & 0.027\\
0.559 & 0.012\\
0.572 & 0.000 \\
\hline
\end{tabular}
}
\hspace{0.5cm}
\parbox{.45\linewidth}{
\centering
\caption{$\theta = \frac{\pi}{8}$}
\begin{tabular}{|| c c ||}
\hline
$\xi$ & $\sharp$ random bits\\
\hline
0.000 & 1.000\\
0.013 & 0.941\\
0.027 & 0.884\\
0.040 & 0.830\\
0.053 & 0.779\\
0.067 & 0.729\\
0.080 & 0.682\\
0.093 & 0.637\\
0.106 & 0.595\\
0.120 & 0.555\\
0.133 & 0.519\\
0.146 & 0.485\\
0.160 & 0.453\\
0.173 & 0.424\\
0.186 & 0.396\\
0.200 & 0.371\\
0.213 & 0.348\\
0.226 & 0.327\\
0.240 & 0.307\\
0.253 & 0.289\\
0.266 & 0.273\\
0.280 & 0.258\\
0.293 & 0.243\\
0.306 & 0.229\\
0.319 & 0.214\\
0.333 & 0.200\\
0.346 & 0.186\\
0.359 & 0.171\\
0.373 & 0.157\\
0.386 & 0.143\\
0.399 & 0.129\\
0.413 & 0.115\\
0.426 & 0.102\\
0.439 & 0.089\\
0.453 & 0.077\\
0.466 & 0.064\\
0.479 & 0.053\\
0.493 & 0.041\\
0.506 & 0.031\\
0.519 & 0.021\\
0.532 & 0.012\\
0.546 & 0.004\\
0.559 & 0.000\\
0.572 & 0.000\\
\hline
\end{tabular}
}
\end{table}


\begin{table}
\parbox{.45\linewidth}{
\centering
\caption{$\theta = \frac{\pi}{16}$}
\begin{tabular}{|| c c ||}
\hline
$\xi$ & $\sharp$ random bits\\
\hline
0.000 & 1.000\\
0.013 & 0.896\\
0.027 & 0.800\\
0.040 & 0.714\\
0.053 & 0.641\\
0.067 & 0.577\\
0.080 & 0.521\\
0.093 & 0.473\\
0.106 & 0.429\\
0.120 & 0.391\\
0.133 & 0.356\\
0.146 & 0.325\\
0.160 & 0.297\\
0.173 & 0.271\\
0.186 & 0.248\\
0.200 & 0.227\\
0.213 & 0.207\\
0.226 & 0.190\\
0.240 & 0.174\\
0.253 & 0.159\\
0.266 & 0.146\\
0.280 & 0.134\\
0.293 & 0.122\\
0.306 & 0.112\\
0.319 & 0.103\\
0.333 & 0.095\\
0.346 & 0.087\\
0.359 & 0.078\\
0.373 & 0.070\\
0.386 & 0.062\\
0.399 & 0.055\\
0.413 & 0.047\\
0.426 & 0.040\\
0.439 & 0.034\\
0.453 & 0.027\\
0.466 & 0.021\\
0.479 & 0.016\\
0.493 & 0.011\\
0.506 & 0.007\\
0.519 & 0.003\\
0.532 & 0.000\\
0.546 & 0.000\\
0.559 & 0.000\\
0.572 & 0.000\\
\hline
\end{tabular}
}
\hspace{0.5cm}
\parbox{.45\linewidth}{
\centering
\caption{$\theta = \frac{\pi}{32}$}
\begin{tabular}{|| c c ||}
\hline
$\xi$ & $\sharp$ random bits\\
\hline
0.000 & 1.000\\
0.013 & 0.823\\
0.027 & 0.706\\
0.040 & 0.619\\
0.053 & 0.551\\
0.067 & 0.493\\
0.080 & 0.444\\
0.093 & 0.400\\
0.106 & 0.362\\
0.120 & 0.328\\
0.133 & 0.297\\
0.146 & 0.269\\
0.160 & 0.244\\
0.173 & 0.221\\
0.186 & 0.200\\
0.200 & 0.181\\
0.213 & 0.163\\
0.226 & 0.147\\
0.240 & 0.133\\
0.253 & 0.119\\
0.266 & 0.107\\
0.280 & 0.095\\
0.293 & 0.085\\
0.306 & 0.076\\
0.319 & 0.067\\
0.333 & 0.059\\
0.346 & 0.052\\
0.359 & 0.046\\
0.373 & 0.040\\
0.386 & 0.035\\
0.399 & 0.030\\
0.413 & 0.025\\
0.426 & 0.021\\
0.439 & 0.017\\
0.453 & 0.013\\
0.466 & 0.009\\
0.479 & 0.006\\
0.493 & 0.004\\
0.506 & 0.002\\
0.519 & 0.000\\
0.532 & 0.000\\
0.546 & 0.000\\
0.559 & 0.000\\
0.572 & 0.000\\
\hline
\end{tabular}
}
\end{table}

\clearpage


\begin{thebibliography}{10}

\bibitem{Acin2012}
Antonio Ac\'in, Serge Massar, and Stefano Pironio.
\newblock Randomness versus nonlocality and entanglement.
\newblock {\em Phys. Rev. Lett.}, 108:100402, Mar 2012.
\newblock \href {http://dx.doi.org/10.1103/PhysRevLett.108.100402}
  {\path{doi:10.1103/PhysRevLett.108.100402}}.

\bibitem{Acin2015}
Antonio Ac{\'i}n, Stefano Pironio, Tam{\'a}s V{\'e}rtesi, and Peter Wittek.
\newblock Optimal randomness certification from one entangled bit.
\newblock {\em Phys. Rev. A}, 93(4):040102, April 2016.
\newblock \href {http://dx.doi.org/10.1103/PhysRevA.93.040102}
  {\path{doi:10.1103/PhysRevA.93.040102}}.

\bibitem{bancal2014more}
Jean-Daniel Bancal, Lana Sheridan, and Valerio Scarani.
\newblock More randomness from the same data.
\newblock {\em New J. Phys.}, 16(3):033011, 2014.
\newblock \href {http://dx.doi.org/10.1088/1367-2630/16/3/033011}
  {\path{doi:10.1088/1367-2630/16/3/033011}}.

\bibitem{Bell1964}
John~S. Bell.
\newblock On the {Einstein Podolsky Rosen} paradox.
\newblock {\em Physics}, 1:195, 1964.

\bibitem{BellReview}
Nicolas Brunner, Daniel Cavalcanti, Stefano Pironio, Valerio Scarani, and
  Stephanie Wehner.
\newblock Bell nonlocality.
\newblock {\em Rev. Mod. Phys.}, 86:419--478, Apr 2014.
\newblock \href {http://dx.doi.org/10.1103/RevModPhys.86.419}
  {\path{doi:10.1103/RevModPhys.86.419}}.

\bibitem{buci}
Orest Bucicovschi and Ji{\v{r}}\'i Lebl.
\newblock On the continuity and regularity of convex extensions.
\newblock {\em J. Convex Anal.}, 20(4):1113--1126, 2013.

\bibitem{Colbeck}
Roger Colbeck.
\newblock {\em Quantum and Relativistic Protocols for Secure Multi-Party
  Computation}.
\newblock PhD thesis, University of Cambridge, 2006.

\bibitem{CR}
Roger Colbeck and Renato Renner.
\newblock Free randomness can be amplified.
\newblock {\em Nat. Phys.}, 8(6):450--454, May 2012.
\newblock \href {http://dx.doi.org/10.1038/nphys2300}
  {\path{doi:10.1038/nphys2300}}.

\bibitem{curchod2017unbounded}
Florian~J. Curchod, Markus Johansson, Remigiusz Augusiak, Matty~J. Hoban, Peter
  Wittek, and Antonio Ac{\'{\i}}n.
\newblock Unbounded randomness certification using sequences of measurements.
\newblock {\em Phys. Rev. A}, 95(2), feb 2017.
\newblock \href {http://dx.doi.org/10.1103/physreva.95.020102}
  {\path{doi:10.1103/physreva.95.020102}}.

\bibitem{DAriano}
Giacomo~Mauro D'Ariano, Paoloplacido~Lo Presti, and Paolo Perinotti.
\newblock Classical randomness in quantum measurements.
\newblock {\em J. Phys. A: Math. Gen.}, 38(26):5979, 2005.
\newblock \href {http://dx.doi.org/10.1088/0305-4470/38/26/010}
  {\path{doi:10.1088/0305-4470/38/26/010}}.

\bibitem{DeLaTorre}
Gonzalo de~la Torre, Matty~J. Hoban, Chirag Dhara, Giuseppe Prettico, and
  Antonio Ac\'in.
\newblock Maximally nonlocal theories cannot be maximally random.
\newblock {\em Phys. Rev. Lett.}, 114(16):160502, 2015.
\newblock \href {http://dx.doi.org/10.1103/physrevlett.114.160502}
  {\path{doi:10.1103/physrevlett.114.160502}}.

\bibitem{Gallego}
Rodrigo Gallego, Lluis Masanes, Gonzalo De~La~Torre, Chirag Dhara, Leandro
  Aolita, and Antonio Ac\'in.
\newblock Full randomness from arbitrarily deterministic events.
\newblock {\em Nat. Commun.}, 4:2654, 2013.
\newblock \href {http://dx.doi.org/10.1038/ncomms3654}
  {\path{doi:10.1038/ncomms3654}}.

\bibitem{GallegoSeq}
Rodrigo Gallego, Lars~Erik W\"urflinger, Rafael Chaves, Antonio Ac\'in, and
  Miguel Navascu\'es.
\newblock Nonlocality in sequential correlation scenarios.
\newblock {\em New J. Phys.}, 16(3):033037, 2014.
\newblock \href {http://dx.doi.org/10.1088/1367-2630/16/3/033037}
  {\path{doi:10.1088/1367-2630/16/3/033037}}.

\bibitem{hu2016experimental}
Meng-Jun Hu, Zhi-Yuan Zhou, Xiao-Min Hu, Chuan-Feng Li, Guang-Can Guo, and
  Yong-Sheng Zhang.
\newblock Experimental sharing of nonlocality among multiple observers with one
  entangled pair via optimal weak measurements.
\newblock {\em arXiv:1609.01863}, Sep 2016.
\newblock \href {http://arxiv.org/abs/1609.01863} {\path{arXiv:1609.01863}}.

\bibitem{Masanes2006}
Lluis Masanes, Antonio Ac\'in, and Nicolas Gisin.
\newblock General properties of nonsignaling theories.
\newblock {\em Phys. Rev. A}, 73(1):012112, Jan 2006.
\newblock \href {http://dx.doi.org/10.1103/physreva.73.012112}
  {\path{doi:10.1103/physreva.73.012112}}.

\bibitem{Nieto}
Olmo Nieto-Silleras, Stefano Pironio, and Jonathan Silman.
\newblock Using complete measurement statistics for optimal device-independent
  randomness evaluation.
\newblock {\em New J. Phys.}, 16(1):013035, 2014.
\newblock \href {http://dx.doi.org/10.1088/1367-2630/16/1/013035}
  {\path{doi:10.1088/1367-2630/16/1/013035}}.

\bibitem{Randomness}
Stefano Pironio, Antonio Ac{\'i}n, Serge Massar, A~Boyer de~La~Giroday,
  Dzimitry~N Matsukevich, Peter Maunz, Steven Olmschenk, David Hayes, Le~Luo,
  T~Andrew Manning, and C.~Monroe.
\newblock Random numbers certified by {Bell}'s theorem.
\newblock {\em Nature}, 464(7291):1021--1024, 2010.
\newblock \href {http://dx.doi.org/10.1038/nature09008}
  {\path{doi:10.1038/nature09008}}.

\bibitem{Popescu}
Sandu Popescu.
\newblock Bell's inequalities and density matrices: Revealing ``hidden''
  nonlocality.
\newblock {\em Phys. Rev. Lett.}, 74(14):2619--2622, Apr 1995.
\newblock \href {http://dx.doi.org/10.1103/physrevlett.74.2619}
  {\path{doi:10.1103/physrevlett.74.2619}}.

\bibitem{PR}
Sandu Popescu and Daniel Rohrlich.
\newblock Quantum nonlocality as an axiom.
\newblock {\em Found. Phys.}, 24(3):379--385, 1994.
\newblock \href {http://dx.doi.org/10.1007/BF02058098}
  {\path{doi:10.1007/BF02058098}}.

\bibitem{rock}
Tyrrell Rockafellar.
\newblock {\em Convex Analysis}.
\newblock Princeton Press, 1970.

\bibitem{schiavon2017three}
Matteo Schiavon, Luca Calderaro, Mirko Pittaluga, Giuseppe Vallone, and Paolo
  Villoresi.
\newblock Three-observer bell inequality violation on a two-qubit entangled
  state.
\newblock {\em Quantum Science and Technology}, 2(1):015010, mar 2017.
\newblock \href {http://dx.doi.org/10.1088/2058-9565/aa62be}
  {\path{doi:10.1088/2058-9565/aa62be}}.

\bibitem{Silva}
Ralph Silva, Nicolas Gisin, Yelena Guryanova, and Sandu Popescu.
\newblock Multiple observers can share the nonlocality of half of an entangled
  pair by using optimal weak measurements.
\newblock {\em Phys. Rev. Lett.}, 114(25):250401, 2015.
\newblock \href {http://dx.doi.org/10.1103/physrevlett.114.250401}
  {\path{doi:10.1103/physrevlett.114.250401}}.

\bibitem{Vazirani}
Umesh Vazirani and Thomas Vidick.
\newblock Certifiable quantum dice.
\newblock {\em Phil. Trans. R. Soc. A.}, 370(1971):3432--3448, Jun 2012.
\newblock \href {http://dx.doi.org/10.1098/rsta.2011.0336}
  {\path{doi:10.1098/rsta.2011.0336}}.

\bibitem{wittek2015ncpol2sdpa}
Peter Wittek.
\newblock Algorithm 950: Ncpol2sdpa---sparse semidefinite programming
  relaxations for polynomial optimization problems of noncommuting variables.
\newblock {\em ACM Trans. Math. Software}, 41(3):21, 2015.
\newblock \href {http://dx.doi.org/10.1145/2699464}
  {\path{doi:10.1145/2699464}}.

\bibitem{yamashita2003implementation}
Makoto Yamashita, Katsuki Fujisawa, and Masakazu Kojima.
\newblock Implementation and evaluation of {SDPA} 6.0 (semidefinite programming
  algorithm 6.0).
\newblock {\em Optim. Method. Softw.}, 18(4):491--505, 2003.
\newblock \href {http://dx.doi.org/10.1080/1055678031000118482}
  {\path{doi:10.1080/1055678031000118482}}.

\end{thebibliography}
\end{document}